\documentclass[sigconf,balance=false]{acmart}
\usepackage{popets}

\usepackage[utf8]{inputenc}
\usepackage{float} 

\usepackage{amsmath,amsthm,amsfonts}
\usepackage{mathtools}
\usepackage{centernot}
\usepackage{svg}
\usepackage{graphicx,xspace}
\usepackage{paralist}
\usepackage{appendix}
\usepackage{xspace}
\usepackage{setspace}
\usepackage{multirow}
\usepackage{ragged2e}
\usepackage{float}
\usepackage[normalem]{ulem}

\usepackage{natbib} 
\usepackage{cleveref}
\usepackage{algorithm}
\usepackage{mathrsfs}
\usepackage[noend]{algpseudocode}
\usepackage{algorithmicx}

\newtheorem{theorem}{Theorem}

\newtheorem{lemma}[theorem]{Lemma}
\newtheorem{definition}[theorem]{Definition}

\newtheorem{corollary}[theorem]{Corollary}

\newtheorem*{definition*}{Definition}

\newcommand{\revision}[1]{#1}

\newcommand{\eps}{\varepsilon}

\newcommand{\universe}{\mathcal{U}}

\newcommand{\mech}{\ensuremath{\mathcal{M}}\xspace}

\newcommand{\thresh}{\ensuremath{T}\xspace}

\newcommand{\maxcontrib}{\ensuremath{\Delta_0}\xspace}
\newcommand{\numiters}{\ensuremath{I}\xspace}
\newcommand{\N}{\ensuremath{\mathcal{N}}}

\newcommand{\privratio}{\ensuremath{r}\xspace}

\newcommand{\reddit}{\texttt{Reddit}\xspace}
\newcommand{\twitter}{\texttt{Twitter}\xspace}
\newcommand{\finance}{\texttt{Finance}\xspace}
\newcommand{\imdb}{\texttt{IMDb}\xspace}
\newcommand{\wikipedia}{\texttt{Wikipedia}\xspace}
\newcommand{\amazon}{\texttt{Amazon}\xspace}
\newcommand{\synth}{\texttt{Synthetic\_80M}\xspace}

\newcommand{\calx}{\ensuremath{\mathcal{X}}}
\newcommand{\caly}{\ensuremath{\mathcal{Y}}}
\newcommand{\calz}{\ensuremath{\mathcal{Z}}}

\newcommand{\noisedistr}{\ensuremath{\mathcal{D}}}

\newcommand{\wthist}{\ensuremath{\texttt{Weighted\_Hist}}}

\newcommand{\naive}{\texttt{Weighted\_Gauss}\xspace}

\setcopyright{popets}
\copyrightyear{YYYY}

\acmYear{YYYY}
\acmVolume{YYYY}
\acmNumber{X}
\acmDOI{XXXXXXX.XXXXXXX}
\acmISBN{}
\acmConference{Proceedings on Privacy Enhancing Technologies}
\begin{document}

\title[DP-SIPS: A simpler, more scalable mechanism for differentially private partition selection]{DP-SIPS: A simpler, more scalable mechanism for differentially private partition selection}

  \author{Marika Swanberg}
  \authornote{Work done while at Tumult Labs}
  \affiliation{
     \institution{Boston University}
     \department{Department of Computer Science}
     \country{}}
  \email{marikas@bu.edu}

  \author{Damien Desfontaines}
     \affiliation{
     \institution{Tumult Labs}
     \country{}}
  \email{damien@desfontain.es}

  \author{Samuel Haney}
     \affiliation{
     \institution{Tumult Labs}
     \country{}}
  \email{sam.haney@tmlt.io}

\begin{abstract}
    Partition selection, or set union, is an important primitive in differentially private mechanism design: in a database where each user contributes a list of items, the goal is to publish as many of these items as possible under differential privacy.
    
    In this work, we present a novel mechanism for differentially private partition selection. This mechanism, which we call {DP-SIPS}, is very simple: it consists of iterating the naive algorithm over the data set multiple times, removing the released partitions from the data set while increasing the privacy budget at each step. This approach preserves the scalability benefits of the naive mechanism, yet its utility compares favorably to more complex approaches developed in prior work.
\end{abstract}

\keywords{Differential privacy, partition selection, scalable algorithms}

\maketitle

\section{Introduction}
\revision{

Group-level aggregation is a fundamental building block for many data analysis tasks. For example, doctors may want to understand mortality rates for patients with different underlying conditions, and search engine developers may want to study the frequency of different search queries that users are making.  

With increased awareness of the risks of data disclosure, increasingly data analysts are using differentially private mechanisms to answer queries on sensitive data sets. In both of the given examples, the aggregate values (e.g., the mortality rates, or the frequencies of search queries) are sensitive, but also the groups themselves are sensitive. Revealing that \emph{some} user in the data set has a particular disease or made a particular search query could \emph{itself} violate privacy---even if it is never revealed \emph{who} that person is. Additionally, the set of groups may be \emph{a priori} unknown or impractical to enumerate, like in vocabulary extraction~\cite{gopi2020differentially} or the private release of search queries~\cite{korolova2009releasing}. 

In order to privately do a group-level aggregation in such settings, one must first privately compute the set of groups represented in the data set---ideally, releasing as many as possible while still maintaining privacy. More formally, the mechanism $M$ gets a data set $x = (W_1, \ldots, W_N)$ where each user $i$ has a list of items $W_i$ and $M$ differentially privately releases a subset $S \subseteq \cup_i W_i$ of the partitions that is as large as possible.

We study this problem, called \emph{differentially private partition selection} (also called key selection or set union), with a focus on approaches that can be incorporated into general-purpose differentially private tooling. When designing such frameworks, scalability is of paramount importance~\cite{wilson2020differentially,amin2022plume,tumultanalyticswhitepaper}: underlying mechanisms must be able to run even when the input or the output is too large to fit on a single machine.
In such cases, the computation itself must also be parallelizable, so it can run across multiple machines and avoid unacceptably long running times.
}

Differentially private partition selection mechanisms have been proposed as early as 2009~\cite{korolova2009releasing}, but recent work has shed a new light on this problem, and proposed alternative approaches that bring significant utility gains~\cite{gopi2020differentially,carvalho2022incorporating}.
To obtain these utility improvements, these newer mechanisms use a \emph{greedy} approach: each user considers what items have been contributed by previous users so far, and ``chooses'' which items to contribute according to a \emph{policy}, chosen carefully to maintain a sensitivity bound.
Unfortunately, these mechanisms do not scale: each user chooses their contribution based on the contribution of all previous users, so the data has to be processed one user after another, and the overall algorithm cannot be parallelized.
This intuition can be verified experimentally: we show that such algorithms eventually time out or run out of memory when the input is very large, even on clusters of multiple machines.
Because these greedy algorithms do not scale, differential privacy tools that need to handle large datasets cannot use these smarter approaches, and instead must rely on the naive algorithm and its underwhelming utility~\cite{privacyonbeam,pipelinedp,tumultanalyticssoftware}.

This raises a natural question: can we achieve the utility benefits of policy-based approaches, while preserving the scalability of more naive approaches?
In this work, we introduce a new approach that combines both benefits: DP-SIPS, short for \emph{scalable, iterative partition selection}.

DP-SIPS relies on a simple idea: rather than having to process the data of each user sequentially, it runs the naive, massively-parallelizable algorithm multiple times, splitting the privacy budget between each step and removing items that were previously discovered. It uses a minuscule amount of privacy budget to publish and then remove the ``heavy hitters''--those which have an enormous amount of weight in the histogram--and allocates the majority of the budget to discovering items that remain after heavy hitters are removed from the dataset. On skewed datasets, the removal of heavy hitters allows users to allocate proportionally more weight in subsequent histograms to less-frequent items while the increased privacy budget on later iterations also \emph{lowers} the threshold for releasing an item.

As we show experimentally on multiple real-world and synthetic datasets, this mechanism has a similar utility to greedy approaches, but scales horizontally: increasing the number of cluster nodes significantly reduces runtime.
This makes it a suitable choice for implementation in general-purpose DP infrastructure with high scalability requirements. 



The rest of this paper is organized as follows:
\begin{itemize}
    \item In Section~\ref{sec:prelim}, we formally define the problem and the building blocks we use for DP-SIPS and its privacy accounting.
    \item In Section~\ref{sec:prior}, we detail existing approaches to differentially private partition selection.
    \item In Section~\ref{sec:sips}, we introduce our algorithm and the proof of its privacy guarantees.
    \item In Section~\ref{sec:experiments}, we report on the experimental evaluation of DP-SIPS.
\item Finally, in Section~\ref{sec:discussion}, we discuss our results, report on some unsuccessful approaches that we tried, and outline directions for future work.
\end{itemize}
\section{Preliminaries}\label{sec:prelim}
A data set $x = (W_1, \ldots, W_N)$ contains a set of user lists $W_i \in \universe^*$. We refer to elements in $\universe$ as \emph{items} or \emph{partitions}, and define partition selection (also called key selection or set union) as follows:

\begin{definition}[Private Partition Selection Problem]

Given a (possibly unbounded) universe $\universe$ of items and a data set $x = (W_1, \ldots, W_N)$ of user lists $W_i \in \universe^*$, an algorithm $\mech$ solves the \emph{private partition selection problem} if it is differentially private, and \mech outputs a set $S\subseteq \cup_i W_i$.
\end{definition}

We begin by presenting the standard notion of differential privacy. Two data sets $x, x'$ are \emph{neighbors} if they differ on one user's list: $x = x' \cup W_{i^*}$. Informally, differential privacy requires that an algorithm's output is distributed similarly on every pair of neighboring data sets.

\begin{definition}[Differential Privacy~\cite{dwork2006our, dwork2006calibrating}] A randomized algorithm $\mech: \universe^* \to \caly$ is $(\eps, \delta)$-differentially private if for every pair of neighboring datasets $x, x' \in \universe^*$ and for all subsets $Y \subseteq \caly$, 
\begin{equation*}
    \Pr[\mech(x) \in Y] \leq e^\eps \cdot \Pr[\mech(x')\in Y] + \delta.
\end{equation*}
\end{definition}

A common variant of differential privacy, called zCDP, is useful for analyzing algorithms that sample noise from a Gaussian distribution (as ours will).
The definition of zCDP uses the R\'enyi Divergence:

\begin{definition}[R\'enyi Divergence] Fix two probability distributions $P$ and $Q$ over a discrete domain $S$. Given a positive $\alpha\neq 1$, R\'enyi divergence of order $\alpha$ of distributions $P$ and $Q$ is
\begin{equation*}
    D_\alpha(P||Q) = \frac{1}{1-\alpha} \log \left(\sum_{s\in S} P(s)^\alpha Q(s)^{1-\alpha}\right).
\end{equation*}
\end{definition}

\begin{definition}[$\rho$-zCDP \cite{bun2016concentrated}]\label{def:pure_zCDP}
A randomized mechanism $\mech:\calx^* \to \caly$ satisfies $\rho$-zCDP if, for all $x, x'\in \calx^*$ differing on a single entry,
\begin{align}
    & D_\alpha(M(x) || M(x')) \leq \rho \cdot \alpha \quad \forall \alpha \in (1,\infty).
\end{align}
\end{definition}

\revision{
This definition can also be relaxed to \emph{approximate} zCDP.

\begin{definition}[Approximate zCDP \cite{bun2016concentrated}]\label{def:approx_zCDP}
A randomized mechanism $M:\calx^n \to \caly$ is $\delta$-approximately $\rho$-zCDP if, for all $x, x'\in \calx^n$ differing on a single entry, there exist events $E=E(M(x))$ and $E'=E'(M(x'))$ such that, for all $\alpha \in (1,\infty)$,
\begin{align}
    & D_\alpha(M(x)|_E || M(x')|_{E'}) \leq \rho \cdot \alpha \quad \text{and} \label{eq:approx_zCDP_1}\\
    & D_\alpha(M(x')|_{E'} || M(x)|_{E}) \leq \rho \cdot \alpha\label{eq:approx_zCDP_2},
\end{align}
and $\Pr[E] \geq 1-\delta$ and $\Pr[E']\geq 1-\delta$.
\end{definition}

Approximate zCDP satisfies composition and post-processing properties.

\begin{lemma}[\cite{bun2016concentrated}, Lemma 8.2]\label{lem:comp}
Let $M: \calx^n \to \caly$ and $M': \calx^n \times \caly \to \calz$. Suppose $M$ satisfies $\delta$-approximate $\rho$-zCDP and for all $y \in \caly$,  $M'(\cdot, y): \calx^n \to \calz$ satisfies $\delta'$-approxmiate $\rho'$-zCDP. Define $M'':\calx^n \to \calz$ by $M''(x) = M'(x, M(x))$. Then, $M''$ satisfies $(\delta + \delta' - \delta\cdot \delta')$-approximate $(\rho + \rho')$-zCDP.
\end{lemma}

Next we show a conversion from approximate zCDP to approximate DP.
We rewrite the conversion lemma from \cite{bun2016concentrated} to be slightly more general (i.e. it uses an existing conversion from zCDP to approximate DP), and then apply the tight zCDP-to-approximate DP conversion given in \cite{canonne2020discrete}.
Overall, this gives a tighter approximate zCDP-to-approximate DP conversion than what is stated in Lemma 8.8 of \cite{bun2016concentrated}.

\begin{lemma}[Generalized version of Lemma 8.8 from \cite{bun2016concentrated}]\label{lem:approx_zCDP_to_DP_general}
  Suppose we can show that every mechanism that satisfies $\rho$-zCDP must satisfy $\epsilon^{*}(\rho), \delta^{*}(\rho)$-approximate DP.
  That is, ($\epsilon^{*}, \delta^{*}$) is a function converting a (pure) zCDP guarantee to an approximate DP guarantee.
  Suppose $\mech: \calx^N \to \caly$ satisfies $\delta$-approximate $\rho$-zCDP.
  Then, $\mech$ satisfies $(\epsilon^{*}(\rho), \delta + (1-\delta)\delta^{*}(\rho))$-DP.
\end{lemma}

\begin{lemma}[\cite{canonne2020discrete}, Proposition 7]\label{lem:zCDP_to_DP_tight}
  Suppose $\mech: \calx^N \to \caly$ satisfies $\rho$-zCDP.
  Then $\mech$ satisfies $(\epsilon, \delta)$-approximate DP for any $\epsilon > 0$ and
\begin{equation}
    \delta = \inf_{\alpha\in (1, \infty)}\frac{\exp((\alpha-1)(\alpha \cdot \rho-\epsilon))}{\alpha-1}\left(1 - \frac{1}{\alpha}\right)^{\alpha}.
\end{equation}
\end{lemma}

Combining Lemma~\ref{lem:approx_zCDP_to_DP_general} and Lemma~\ref{lem:zCDP_to_DP_tight} gives us the following.

\begin{corollary}\label{cor:approx_zCDP_to_DP}
  Suppose $\mech: \calx^N \to \caly$ satisfies $\delta$-approximate $\rho$-zCDP.
  Then $\mech$ satisfies $(\epsilon, \delta + (1-\delta)\delta')$-approximate DP for any $\epsilon > 0$ and
\begin{equation}
    \delta' = \inf_{\alpha\in (1, \infty)}\frac{\exp((\alpha-1)(\alpha\cdot\rho-\epsilon))}{\alpha-1}\left(1 - \frac{1}{\alpha}\right)^{\alpha}.
\end{equation}
\end{corollary}
}

A common primitive in building private algorithms, the Gaussian Mechanism, satisfies $\rho$-zCDP.

\begin{definition}[Gaussian Distribution] The Gaussian distribution with parameter $\sigma$ and mean 0, denoted $\N(0,\sigma^2)$ is defined for all $\ell \in \mathbb{R}$ and has probability density
\begin{equation*}
    h(\ell) = \frac{1}{\sigma \sqrt{2\pi}} e^{-\frac{\ell^2}{2\sigma^2}}.
\end{equation*}
\end{definition}

\begin{definition}[$\ell_2$-Sensitivity] Let $f: \universe^n \to \mathbb{R}^d$ be a function. Its $\ell_2$-sensitivity is
\begin{equation*}
    \Delta_f = \max_{
    \substack{x, x'\in \universe\\
    x, x' neighbors}
    } \| f(x) - f(x') \|_2.
\end{equation*}
\end{definition}

\begin{definition}[Gaussian Mechanism~\cite{bun2016concentrated}]\label{def:gauss_mech}
Let $f:\universe^n \to \mathbb{R}^d$ be a function with $\ell_2$-sensitivity $\Delta_f$. Then the Gaussian mechanism is the algorithm
\begin{equation*}
    \mech_f(x) = f(x) + (Z_1, \ldots, Z_d),
\end{equation*}
where $Z_i \sim \N\left(0, \frac{\Delta_f^2}{2\rho}\right)$. Algorithm $\mech_f$ satisfies $\rho$-zCDP.
\end{definition}


\section{Prior Approaches to Partition Selection}\label{sec:prior}

In this section we discuss three existing algorithms for differentially private partition selection. We begin with the naive algorithm, called Weighted Gaussian, in Section~\ref{sec:naive}. In Section~\ref{sec:greedy_algs}, we present Policy Gaussian~\cite{gopi2020differentially} and Greedy updates Without sampling~\cite{carvalho2022incorporating}.

The algorithms all have three main steps: first, they compute a \emph{weighted histogram}, which is simply a mapping from an item $u \in \universe$ to a weight $H[u] \in \mathbb{R}$; second, they add calibrated noise to each item in the histogram; and lastly, they release items that are above some appropriately-chosen threshold $\thresh$. The primary difference between the algorithms is in how they compute the weighted histogram. This high-level algorithm for private partition selection is described in Algorithm~\ref{alg:high_level}, which can be composed with different weighted histogram algorithms.

\begin{algorithm}
    \caption{High-Level Partition Selection Algorithm}
    \label{alg:high_level}
    \begin{FlushLeft}
    \textbf{Input:} Data set of user partitions $x= (W_1, \ldots, W_{N})$\\
    \hspace{3em}Weighted Histogram Algorithm \wthist \\
    \hspace{3em}Threshold \thresh\\
    \hspace{3em}Noise distribution \noisedistr\\ 
    \textbf{Output:} Partitions $S \subseteq \cup_i W_i$
     \end{FlushLeft}
    \begin{algorithmic}[1] 
            \State Initialize empty set $S \gets \{\}$
            \State $H \gets \wthist(x)$ \Comment{Compute weighted histogram}
            \For{$u\in Supp(H)$} 
                \State $Z_u \sim \noisedistr$
                \State $\hat{H}[u]\gets H[u] + Z_u$
                \If{$\hat{H}[u] \geq \thresh$}
                    \State $S \gets S \cup \{u\}$
                \EndIf
            \EndFor
            \Return $S$
    \end{algorithmic}
\end{algorithm}

\subsection{Baseline: Weighted Gaussian}\label{sec:naive}

The Weighted Gaussian algorithm (Algorithm~\ref{alg:naive}) is one of the simplest and first algorithms for private partition selection \cite{korolova2009releasing}. To build a weighted histogram, the algorithm first pre-processes the data set by removing any duplicates within a user's set and truncating each user's set to have at most $\maxcontrib$ items. Then, the users compute a histogram with bounded $\ell_2$-sensitivity as follows: each user $i$ updates the weight $H[u]$ for each of the items $u$ in their set $W_i$ with the following rule:
\begin{equation*}
    H[u] \gets H[u] + \frac{1}{\sqrt{|W_i|}}.
\end{equation*}
 The resulting weighted histogram has an $\ell_2$-sensitivity of 1, so it can be composed with the high-level algorithm using calibrated Gaussian noise and an appropriate threshold to ensure that the overall algorithm satisfies $\delta$-approximate $\rho$-zCDP. In the statement of Algorithm~\ref{alg:naive}, $\Phi(\cdot)$ is used to denote the cumulative density function of the standard Gaussian distribution and $\Phi^{-1}(\cdot)$ is its inverse.
 
 \begin{algorithm}
    \caption{Weighted Gaussian}
    \label{alg:naive}
    \begin{FlushLeft}
    \textbf{Input:} Data set $x= (W_1, \ldots, W_{N})$\\
    \hspace{3em}Privacy parameters $(\rho, \delta)$\\
    \hspace{3em}Maximum per-user contribution \maxcontrib\\
    \textbf{Output:} Partitions $S \subseteq \cup_i W_i$
     \end{FlushLeft}
    \begin{algorithmic}[1] 
            \State Initialize empty histogram $H \gets \{\}$
            \State Initialize empty set $S \gets \{\}$
            \For {$i=1, \ldots, N$} 
                \State $W_i \gets$ get rid of duplicate items from $W_i$
                \State $\overline{W_i} \gets$ uniformly sample at most \maxcontrib items from $W_i$
                \For {$u \in \overline{W}_i$} 
            \State $H[u] \gets H[u] + \frac{1}{\sqrt{|\overline{W}_i|}}$
                \EndFor
            \EndFor
            \State $\sigma \gets \frac{1}{\sqrt{2\rho}}$
            \State $\thresh \gets \max_{k \in [\maxcontrib]} \left\{ \frac{1}{\sqrt{k}} + \sigma \cdot \Phi^{-1}\left((1-\delta)^{1/k}\right) \right \} $
            \For{$u\in Supp(H)$} \
                \State $\hat{H}[u]\gets H[u] + \N(0, \frac{1}{2\rho})$
                \If{$\hat{H}[u] \geq \thresh$}
                    \State $S \gets S \cup \{u\}$
                \EndIf
            \EndFor
            \Return $S$
    \end{algorithmic}
\end{algorithm}


\begin{theorem}\label{thm:naive_priv}
Fix any $\rho >0$, any $\delta \in (0,1)$, and any $\maxcontrib \in \mathbb{N}$. The Weighted Gaussian algorithm (Algorithm~\ref{alg:naive}) satisfies $\delta$-approximate $\rho$-zCDP.
\end{theorem}

See Appendix~\ref{app:naive_priv} for the proof of Theorem~\ref{thm:naive_priv}.

The Weighted Gaussian algorithm benefits from being highly scalable. In particular, it lends itself well to parallel computation across several computers within a cluster, because the weights on each histogram item can be computed in parallel as well as the noise addition and thresholding steps (see Figure~\ref{fig:parallelism}). Thus, Algorithm~\ref{alg:naive} is the standard approach for doing partition selection on data sets that are too large to fit in a single machine's memory. Unfortunately, Weighted Gaussian suffers from poor accuracy compared to the greedy approaches we discuss next.

\subsection{Greedy Approaches}\label{sec:greedy_algs} 
One problem with the Weighted Gaussian algorithm is users waste their sensitivity budget on histogram items that are already well above the threshold. Most real-world data has highly skewed item frequencies, but Weighted Gaussian increments all items in a user's set by the same amount. 

The two greedy algorithms we discuss next solve this problem by iterating through the users one-by-one and using an \emph{update policy} and the current state of the histogram to decide how to allocate weight across the items in their set. That is, each user's update depends on previous users' updates. For example, in both algorithms, users do not contribute to items in $H$ that have already reached $\thresh^*$, the threshold $\thresh$ plus some positive buffer; items that have reached the buffered threshold are very likely to be returned  after noise is added and thus do not need more weight. These adaptive update rules are carefully chosen so that the overall algorithm has bounded global sensitivity. As we will discuss in later sections, the main downside of these algorithms is their sequential nature. By design, they require iterating over the entire data set, which may be prohibitively slow for industrial data sets.


\subsubsection{Policy Gaussian~\cite{gopi2020differentially} \revision{(DPSU)}} As with the Weighted Gaussian algorithm, the data set is pre-processed by removing duplicate items from each user's set and truncating the user sets to some fixed maximum size $\maxcontrib$. Then, the algorithm iterates sequentially over the users and, for each item $u$ in user $i$'s set $W_i$, the user increments $H[u]$ by a weight that is proportional to $\thresh^* - H[u]$, where $\thresh^*$ is equal to $\thresh$ plus some positive buffer. Essentially, \emph{items that are further from the buffered threshold get more weight added to them, and those that have already reached the buffered threshold get none}. The weight that a user adds to each item is normalized so a single user's update to $H$ has an $\ell_2$-norm of at most 1.\revision{We refer to this algorithm as DPSU.}

Gopi et al.~prove that the global $\ell_2$-sensitivity of the entire Policy Gaussian algorithm is bounded by 1, so applying the high-level algorithm (Algorithm~\ref{alg:high_level}) with appropriately-scaled Gaussian noise to each item and thresholding are sufficient for satisfying differential privacy. 


\subsubsection{Greedy updates Without sampling (GW)~\cite{carvalho2022incorporating}} Carvalho et al.~observe that, rather than removing duplicate items in a user's set, one can use this frequency information to decide where to allocate sensitivity budget. 
GW iterates over the users, and each user computes $u^*$: the most frequent item in their list such that $H[u^*]$ is below the buffered threshold $\thresh^*$. Then, the user increments $H[u^*]$ by  $\min(1, \thresh^* - H[u^*], budget)$ where $budget$ is the user's remaining $\ell_1$ budget. This process is repeated until the user's initial budget of 1 is consumed or the user has no items left that are below $\thresh^*$\footnote{As the algorithm is presented in \cite{carvalho2022incorporating}, it does not always terminate. In particular, they do not consider the case when a user has budget left but all of its items have reached $\thresh^*$, but adding this condition luckily does not affect the privacy proof.}. Carvalho et al.~prove that their GW algorithm for building a weighted histogram has a global $\ell_1$-sensitivity bound of 1, so running the high-level algorithm (Algorithm~\ref{alg:high_level}) with Laplace noise and thresholding is sufficient for satisfying differential privacy. 

One notable feature of their algorithm is that it can use item frequency information from a public data set to increase the accuracy of the frequency estimates. We do not use this feature when comparing it to other algorithms, since our goal is to create a general-purpose algorithm that could be incorporated into a privacy framework without requiring a data analyst to input a public data set.

\section{DP-SIPS}\label{sec:sips}

\begin{figure*}
    \centering
    \includegraphics[page=2, trim={0cm 0cm 0cm 9.55cm},clip, width=469pt]{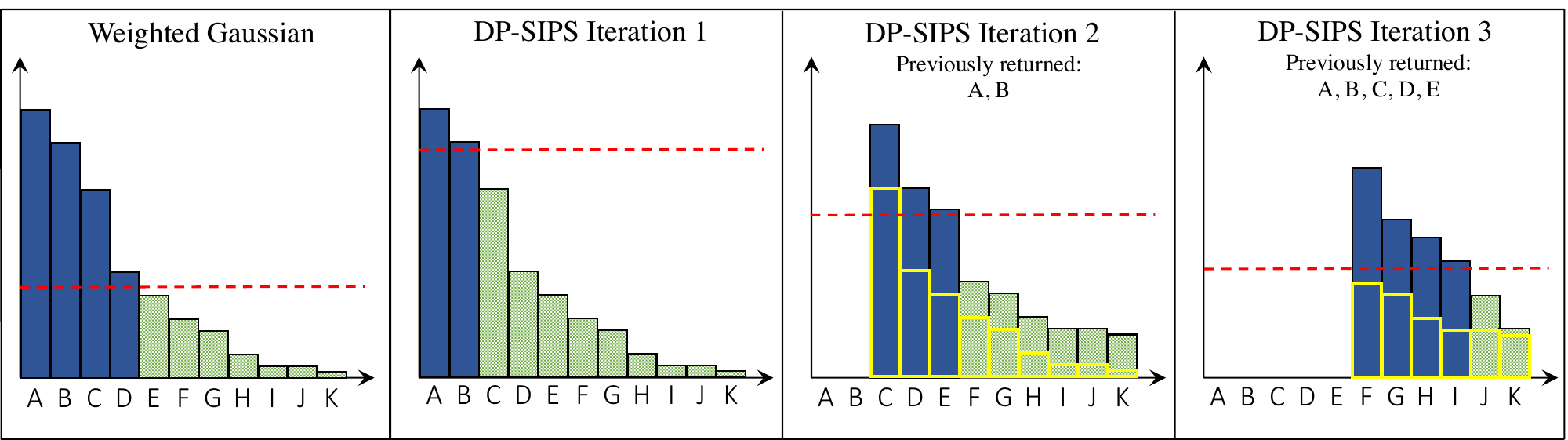}
    \caption{Depiction of Weighted Gaussian noisy histogram (left) compared to intermediate SIPS noisy histograms (right three diagrams) on a skewed data set. Solid blue bars represent partitions that will be returned, and yellow outlines represent the weight of each item on the previous iteration. Although the threshold for Weighted Gaussian is lower than each of the SIPS thresholds, SIPS benefits from less-frequent items getting increased weight as returned partitions are removed from user sets.}
    \label{fig:histograms}
\end{figure*}

In this section we present DP-SIPS: Differentially Private Scalable, Iterative Partition Selection, detailed in Algorithm~\ref{alg:1}. The basic structure is quite simple: it runs Weighted Gaussian on the data set multiple times with increasing privacy budget (and corresponding decreasing thresholds); on each iteration, the partitions returned by Weighted Gaussian are removed from each of the users' sets.

Because Weighted Gaussian updates the histogram uniformly across a user's items, when returned partitions are removed from the user's set, they can allocate more weight to their remaining items (see Figure~\ref{fig:histograms}). The first iteration returns the very popular items using only a tiny fraction of the overall privacy budget, and subsequent iterations yield less frequent items. The action of the algorithm is twofold: \emph{in each iteration, the threshold is lowered at the same time as users allocate more weight to each of the items that remain in their sets} (since the user sets get smaller when previously-returned items are removed).

Furthermore, the user sets are re-truncated on each iteration after the previously returned partitions are removed. For data sets that are both skewed in the item frequencies and in the sizes of users' sets, this allows additional items to be included on each iteration.


\begin{algorithm}
    \caption{DP-SIPS: Scalable, Iterative Partition Selection}
    \label{alg:1}
    \begin{FlushLeft}
    \textbf{Input:} 
    Data set $x= (W_1, \ldots, W_{N})$\\
    \hspace{3em}Maximum per-user contribution \maxcontrib\\
    \hspace{3em}Privacy parameters $\rho, \delta$\\
    \hspace{3em}Number of iterations $I$\\
    \hspace{3em}Privacy budget allocation factor $r>0$\\
    \textbf{Output:} Subset $S \subseteq \cup_i W_i$
    \end{FlushLeft}
    \begin{algorithmic}[1] 
            \State $S \gets\{\}$ 
            \For {$i=0, \ldots, I-1$}
                \revision{\If{$r=1$}
                    \State $(\rho_i, \delta_i) \gets \left(\frac{\rho}{I}, \frac{\delta}{I} \right)$ 
                \Else 
                    \State $(\rho_i, \delta_i) \gets \left(\rho \cdot r^{I-i-1} \cdot \frac{1-r}{1-r^I}, \quad \delta \cdot r^{I-i-1} \cdot \frac{1-r}{1-r^I} \right)$ 
                \EndIf
                }
                \State $P_i \gets \naive\left(x=(W_1, \ldots, W_N), \rho_i, \delta_i, \maxcontrib\right)$
                \For {$j \in N$}
                    \State $W_j \gets W_j \setminus P_i$ \Comment{Remove already-found partitions}
                \EndFor
                \State $S \gets S \cup P_i$
            \EndFor
            \Return $S$
    \end{algorithmic}
\end{algorithm}

\begin{theorem}
    For any $\rho > 0$ and any $\delta \in (0,1]$, Algorithm~\ref{alg:1} satisfies $\delta$-approximate $\rho$-zCDP.
\end{theorem}

\begin{proof}
    By Theorem~\ref{thm:naive_priv}, each call to \naive satisfies $\delta_i$-approximate $\rho_i$-zCDP. Applying composition and postprocessing (\revision{Lemma~\ref{lem:comp}}), Algorithm~\ref{alg:1} satisfies $\left(\sum_{i=0}^{I-1} \delta_i\right)$-approximate $\left(\sum_{i=0}^{I-1} \rho_i\right)$-zCDP.
    
    \revision{If $r=1$, then clearly $\sum_{i=0}^{I-1} \rho/I = \rho$ and similarly $\sum_{i=0}^{I-1} \delta/I = \delta$.}
    
    Now, for $r\neq 1$ we will solve for these summations using the closed-form formula for geometric sums. 
    \begin{align*}
        \sum_{i=0}^{I-1} \rho_i
         & = \sum_{i=0}^{I-1} \rho \cdot r^{I-i-1} \cdot \frac{1-r}{1-r^I} 
          = \rho \cdot \frac{1-r}{1-r^I} \cdot \sum_{i=0}^{I-1}  r^{I-i-1} \\
         & = \rho \cdot \frac{1-r}{1-r^I} \cdot \sum_{j=0}^{I-1}  r^j 
          = \rho \cdot \frac{1-r}{1-r^I} \cdot \frac{1-r^I}{1-r} 
          = \rho.
    \end{align*}
    An identical calculation holds for $\sum_{i=0}^{I-1} \delta_i$. Thus, Algorithm~\ref{alg:1} satisfies $\delta$-approximate $\rho$-zCDP.
\end{proof}

A key advantage of DP-SIPS is that it is massively parallelizable: within each iteration, the mechanism performs a simple groupby-count operation, which can be distributed among multiple cores or a cluster of multiple machines.
DP-SIPS requires multiple iterations, so it takes longer to run than the naive algorithm.
However, only a few iterations are required to achieve good accuracy (see Section~\ref{sec:hyperparameters}), so its runtime is still reasonable even on large datasets.
A visual explanation of the scalability property of DP-SIPS in comparison with other approaches can be found in Figure~\ref{fig:parallelism}, and an experimental performance and scalability comparison can be found in Section~\ref{sec:scalability}.

\revision{

\paragraph{Relation to Algorithm 7 in \cite{kamath2019privately}}

Private Product-Distribution Estimator in \cite{kamath2019privately} privately estimates a product of $d$ Bernoulli distributions, and has some structural similarities to DP-SIPS: it also uses weighted Gaussian iteratively with decreasing thresholds, removing items above the threshold in each round to estimate the frequencies of each coordinate. Aside from the problem setting (bounded vs. unbounded domain) and goal (distribution estimation vs. partition selection), a major difference is that \cite{kamath2019privately} partitions the dataset horizontally into $\log(d/2)$ groups, and uses different subsets of user records on each round. This partitioning step is key to their privacy and accuracy arguments and leads to $\log(d/2)$ rounds rather than the constant number of rounds that DP-SIPS uses.

}

\revision{

\subsection{Scalability Analysis in the MapReduce Framework}

We analyze DP-SIPS in the MapReduce model of distributed computing to better understand its scaling behavior on a cluster. The model proceeds in rounds that are synchronized across worker nodes running in parallel. The data records begin arbitrarily partitioned among the worker nodes. Each round begins with a \emph{map} phase in which the worker node applies a map to each record in its partition. Next, the records are \emph{shuffled} among the workers so that records with the same key are located on the same worker node. Lastly, each worker processes the records in the \emph{reduce} phase. The three phases (map, shuffle, and reduce) constitute one round in the MapReduce framework. Typically, the shuffle phase is the most time intensive since large amounts of data must move among the worker nodes, so the overall number of rounds usually reflects an algorithm's efficiency on a real system.

DP-SIPS proceeds in the following steps (see Figure~\ref{fig:mapreduce}):

\begin{enumerate}
    \item Shuffle records by \texttt{user\_id} \label{item:returnto}
    \item Bound the number of contributions per user
    \item Shuffle records by \texttt{item}. Call this dataset $D$.
    \item Count the number of items, add noise, and store the set of items $S$ above the threshold 
    \item Do a join to remove items in $S$ from $D$. The records now consist of $D\setminus S$. If this is the last iteration, this step can be skipped.
    \item Return to Step \ref{item:returnto} for $I$ iterations in total.
\end{enumerate} 

All but the last iteration of DP-SIPS involves three MapReduce rounds: first shuffling records by user to truncate; then shuffling by item to count, add noise, and threshold; lastly, doing a join with the original data to remove items that were released. In the last iteration, the join is not necessary. So, in total the algorithm does $3(I-1) + 2$ rounds in the MapReduce framework, where $I$ is the number of iterations of DP-SIPS. In our experiments we set $I=3$, so DP-SIPS performs 8 MapReduce rounds. Weighted Gaussian only uses 2 MapReduce rounds, since $I=1$. The greedy algorithms cannot be parallelized--specifically the greedy methods for histogram computation--so each record must be loaded into the one worker node to make its contribution to the histogram. This requires $O(N)$ rounds of MapReduce, where $N$ is the number of users. Indeed, our experimental results in Section~\ref{sec:scalability} (Figures~\ref{tab:scalability} and \ref{tab:scalability2}) show that the runtime of DP-SIPS decreases with the number of cores while the greedy algorithms stay constant.

\begin{figure}
    \centering
    \includegraphics[width=241pt]{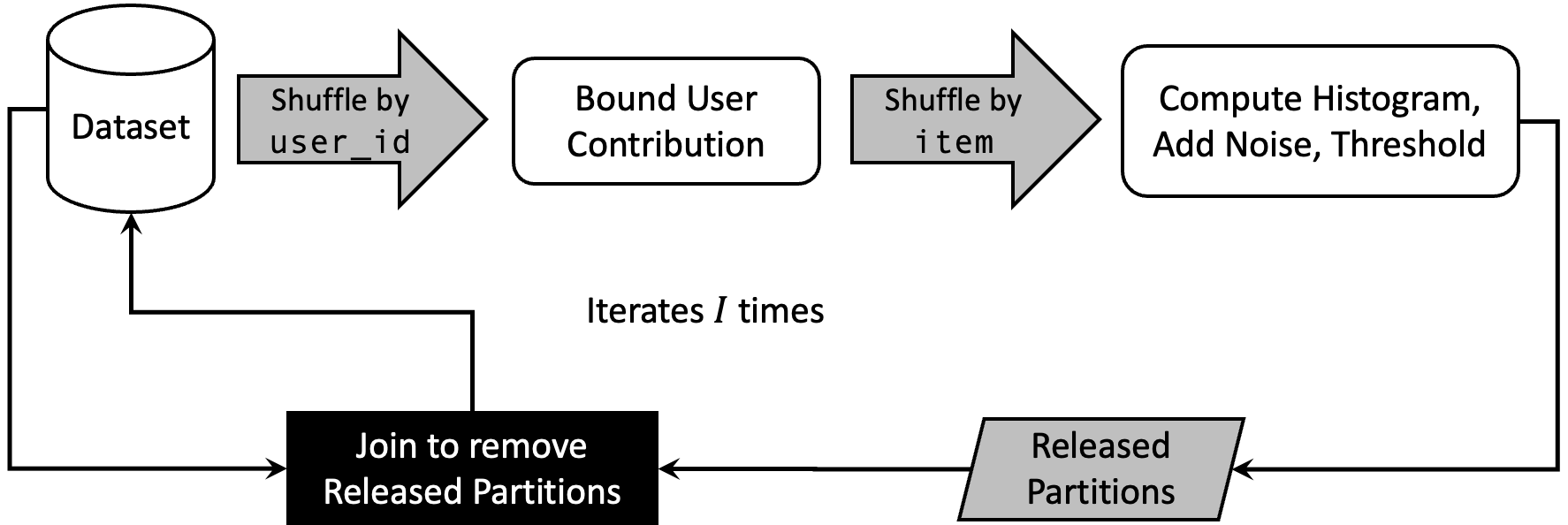}
    \caption{MapReduce Diagram for DP-SIPS. Each iteration within SIPS involves two rounds of MapReduce.}
    \label{fig:mapreduce}
\end{figure}

}


\section{Experimental Results}\label{sec:experiments}

We use data sets of several different sizes from varied text domains to validate the empirical performance of the DP-SIPS algorithm, and we find that in general DP-SIPS has comparable accuracy to DPSU and GW while scaling to large data sets that DPSU and GW timeout on. As with prior works, we focus on the problem of vocabulary discovery under the constraint of user-level privacy. 

We begin by describing the data sets in Section~\ref{sec:datasets} and then, in Section~\ref{sec:accuracy} we discuss how the empirical accuracy of DP-SIPS compares to that of existing algorithms, and in Section~\ref{sec:scalability} we discuss the results from our scalability experiments.

\subsection{Data sets}\label{sec:datasets}

We use six publicly-available data sets and one synthetically-generated data set to study the accuracy of DP-SIPS against existing partition selection algorithms, and we use the largest two of the seven for scalability experiments on Amazon Elastic Map Reduce clusters. \reddit~\cite{reddit_dataset} is a data set of text posts collected from r/AskReddit which appeared as a benchmark in \cite{gopi2020differentially, carvalho2022incorporating}. We also use four data sets that appeared in \cite{carvalho2022incorporating}: \twitter~\cite{twitter_dataset}, comprising customer support tweets to and from large corporations; \finance~\cite{finance_dataset}, financial headlines for stocks; \imdb~\cite{imdb_dataset}, a set of movie reviews scraped from IMDb; \wikipedia~\cite{wiki_dataset}, a set of Wikipedia abstracts (where we treat each abstract as a separate user set). 

For the scalability experiments, we use \amazon~\cite{amazon_dataset}, a publicly-available text data set from Kaggle of 4 million Amazon product reviews, and a synthetically generated data set \synth with 80 million users and 4.6 billion observations.
The synthetic data set, \synth, is generated as follows.
First, each user draws the number of items in their set from a Pareto distribution of scale 10 and shape 1.16; this captures the common ``80-20 principle'' observed on many empirical datasets that states that 80\% of outcomes are due to 20\% of causes.
Then, each item in each user's set is generating by sampling from a zeta distribution of parameter $1.1$; this captures Zip's law, which states that in many types of data (including words in natural languages), the rank-frequency distribution follows an inverse relation.

Table~\ref{tab:dataset_stats} lists the number of users, number of observations, and the vocabulary size of each data set.
Using the same methods as~\cite{gopi2020differentially, carvalho2022incorporating}, we preprocess all non-synthetic data sets using tokenization with \texttt{nltk.word\_tokenize}, removing URLs and symbols, and lower casing all words.

\begin{table}
    \centering
    \begin{tabular}{|c|c|c|c|}
        \hline
        Data set & Users & Observations & Vocabulary Size\\ 
        \hline
        \reddit & 223,388 & 373,983 & 155,701 \\
        \twitter & 702,682 & 2,811,774 & 1,300,123 \\
        \finance & 1,400,465 & 1,400,465  & 267,256 \\
        \imdb & 49,999 & 49,999 & 194,532\\
        \wikipedia & 245,103 & 245,103 & 631,866 \\
        \amazon & 4,000,000 & 4,000,000 & 4,250,427 \\
        \synth & 80,000,000 &4,643,596,660 & 741,129,124 \\
        \hline
    \end{tabular}
    \caption{Number of users, number of observations, and true vocabulary size for each data set we consider.}
    \label{tab:dataset_stats}
\end{table}

\subsection{Accuracy results}\label{sec:accuracy}
Because the goal of partition selection is to privately output as many partitions as possible, we measure accuracy as the \emph{number of partitions released}. To date, there are no analytical accuracy guarantees for any prior partition selection algorithms (in the setting where users contribute multiple items), so we must use experimental validation to understand the accuracy of each algorithm.

We test the accuracy of the four algorithms: Weighted Gaussian (Algorithm~\ref{alg:naive}), SIPS (our Algorithm~\ref{alg:1}), Policy Gaussian (DPSU) from \cite{gopi2020differentially}, and Greedy updates Without sampling (GW) from \cite{carvalho2022incorporating}.

 Table~\ref{tab:accuracy_benchmarks} shows the following accuracy trends for SIPS on the datasets described in Section~\ref{sec:datasets}:
\begin{itemize}
    \item \revision{SIPS only performs slightly worse than both GW and DPSU on one dataset (\finance); in general, SIPS’ performance is on par with DPSU’s.}
    \item The accuracy of SIPS is consistently approximately double that of Weighted Gaussian.
\end{itemize}

\begin{table*}
    \centering
    \begin{tabular}{|c|cccc|}
        \hline
         Data set & Wt. Gauss & \textbf{SIPS} & DPSU & GW \\
         \hline 
         \reddit & 6,160 & 11,392 & 11,186 & 11,984 \\
         \twitter & 12,632 & 23,649 & 23,576 & 27,184 \\
         \finance & 17,350 & 27,559 & 29,005 & 37,503 \\
         \imdb & 3,728 & 7,759 & 5,845 & 3,133\\
         \wikipedia & 11,340 & 21,037 & 18,129 & 11,251\\
         \amazon & 67,522 & 144,805 & 143,997 & 185,563\\
         \synth & 711,601 & 1,137,467 & - & -\\
         \hline
         \finance (deduplicated) & - & - & - & 37,563 \\
         \imdb (deduplicated) & - & - & - & 3,005 \\
         \wikipedia (deduplicated) & - & - & - & 9,802 \\
         \hline
    \end{tabular}
    \vspace{3mm}
    \caption{Number of partitions returned by Wt. Gauss (Algorithm~\ref{alg:naive}), SIPS (Algorithm~\ref{alg:1}), DPSU (Policy Gaussian from~\cite{gopi2020differentially}), and GW (from~\cite{carvalho2022incorporating}) on  six data sets. Additionally, we run GW on three data sets where duplicates within user lists have been removed. Note that the other three algorithms already deduplicate user sets. For SIPS we use 3 iterations and $\privratio=1/3$. For Wt. Gaussian, SIPS, and DPSU, the privacy budget is set to $\rho = 0.1, \delta = 10^{-5}$, and the user contributions are truncated to $\maxcontrib=100$. For GW, $\eps = 1.7$ and $\delta = 8.1142 \times 10^{-5}$\revision{, which implies $10^{-5}$-approximate 0.1-zCDP}. On \synth, DPSU and GW ran out of memory before completing the computation.}
    \label{tab:accuracy_benchmarks}
\end{table*}

In addition, Figures~\ref{fig:vary_eps_reddit}, \ref{fig:vary_delta_reddit}, and \ref{fig:vary_eps_imdb} demonstrate that the relative accuracy of each algorithm on a given data set remains consistent across different choices of privacy parameters.

\subsubsection{\revision{Accuracy Inconsistencies}}

While GW tends to perform well, its accuracy on the \imdb and \wikipedia data sets is below even that of Weighted Gaussian (see Figure~\ref{fig:vary_eps_imdb}). To further investigate this phenomenon, we modify these two data sets in addition to \finance to remove repeated items within each user's set (see the data sets marked as ``deduplicated'' in Table~\ref{tab:accuracy_benchmarks}). Note also that the accuracy of the other three algorithms is unaffected by deduplication since they all perform this preprocessing step to the data sets.

Without any item frequency information, the algorithm adds all of its weight to a randomly-selected item, so one would expect the performance on all three data sets to be worse than Weighted Gaussian. To the contrary, GW's accuracy on the deduplicated \finance data set is still significantly higher than the others, and GW's accuracy on the deduplicated \imdb and \wikipedia is again worse than Weighted Gaussian. Figure~\ref{fig:vary_eps_imdb} shows that GW's poor relative accuracy on \imdb worsens at higher levels of epsilon.

\revision{We believe that GW's inconsistent accuracy depends on the ratio between vocabulary size and number of observations: GW performs very well on datasets where this ratio is small (\finance: 0.19), and very poorly on datasets with large relative vocabulary (\imdb: 3.8). Other factors, such as the short length of \finance headlines compared to \imdb reviews and \wikipedia abstracts, might also play a role. We did not attempt to fully resolve this discrepancy in GW's accuracy};
 however, we present these results as a caution when selecting a partition selection algorithm.

\begin{figure}
    \centering
    \includegraphics[width=241pt]{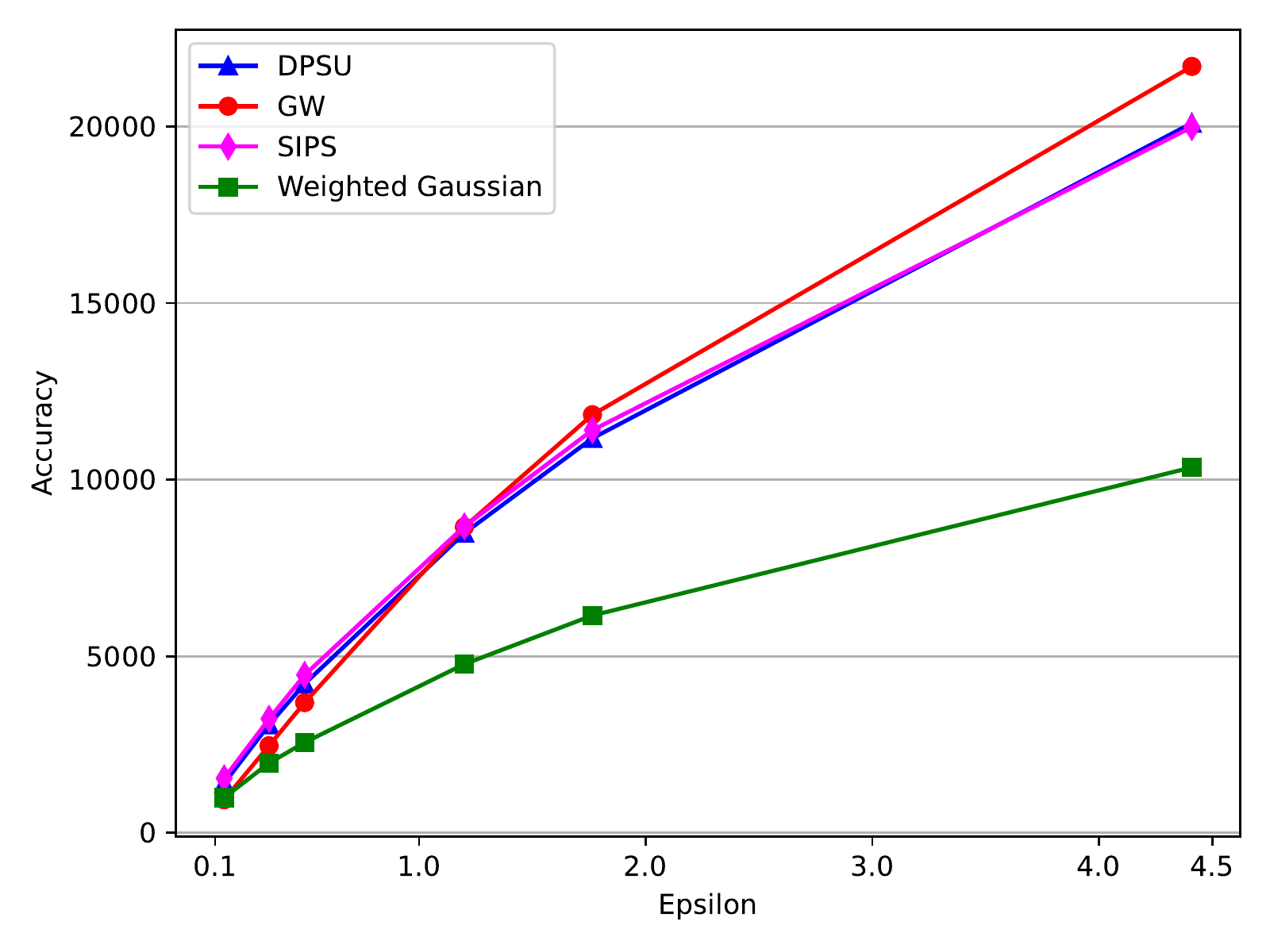}
    \caption{Accuracy on \reddit as a function of varying epsilon (and associated $\rho$), for $\delta=10^{-5}$.}
    \label{fig:vary_eps_reddit}
\end{figure}

\begin{figure}
    \centering
    \includegraphics[width=241pt]{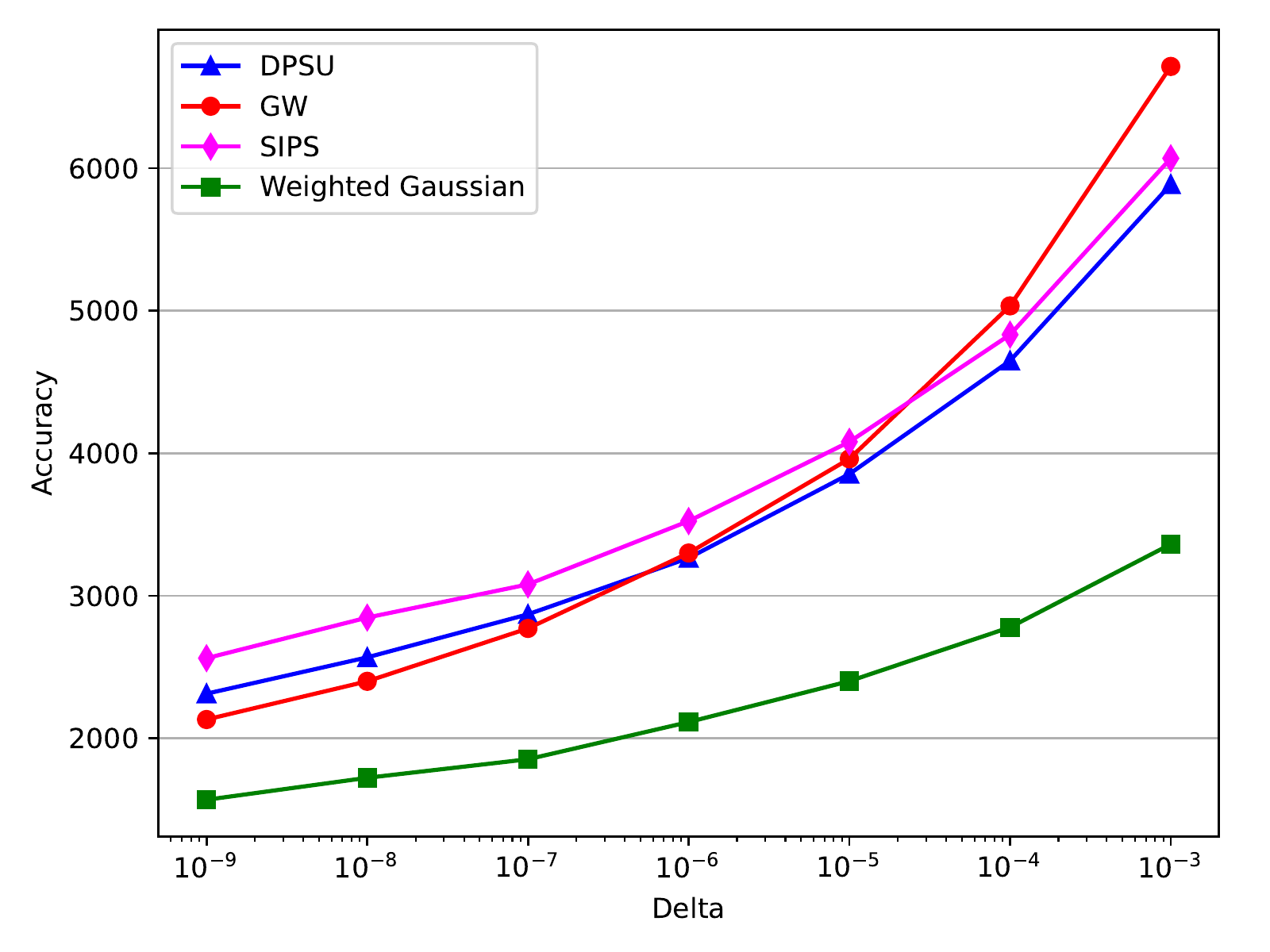}
    \caption{Accuracy on \reddit as a function of varying delta for $\eps = 1.7$.}
    \label{fig:vary_delta_reddit}
\end{figure}

\begin{figure}
    \centering
    \includegraphics[width=241pt]{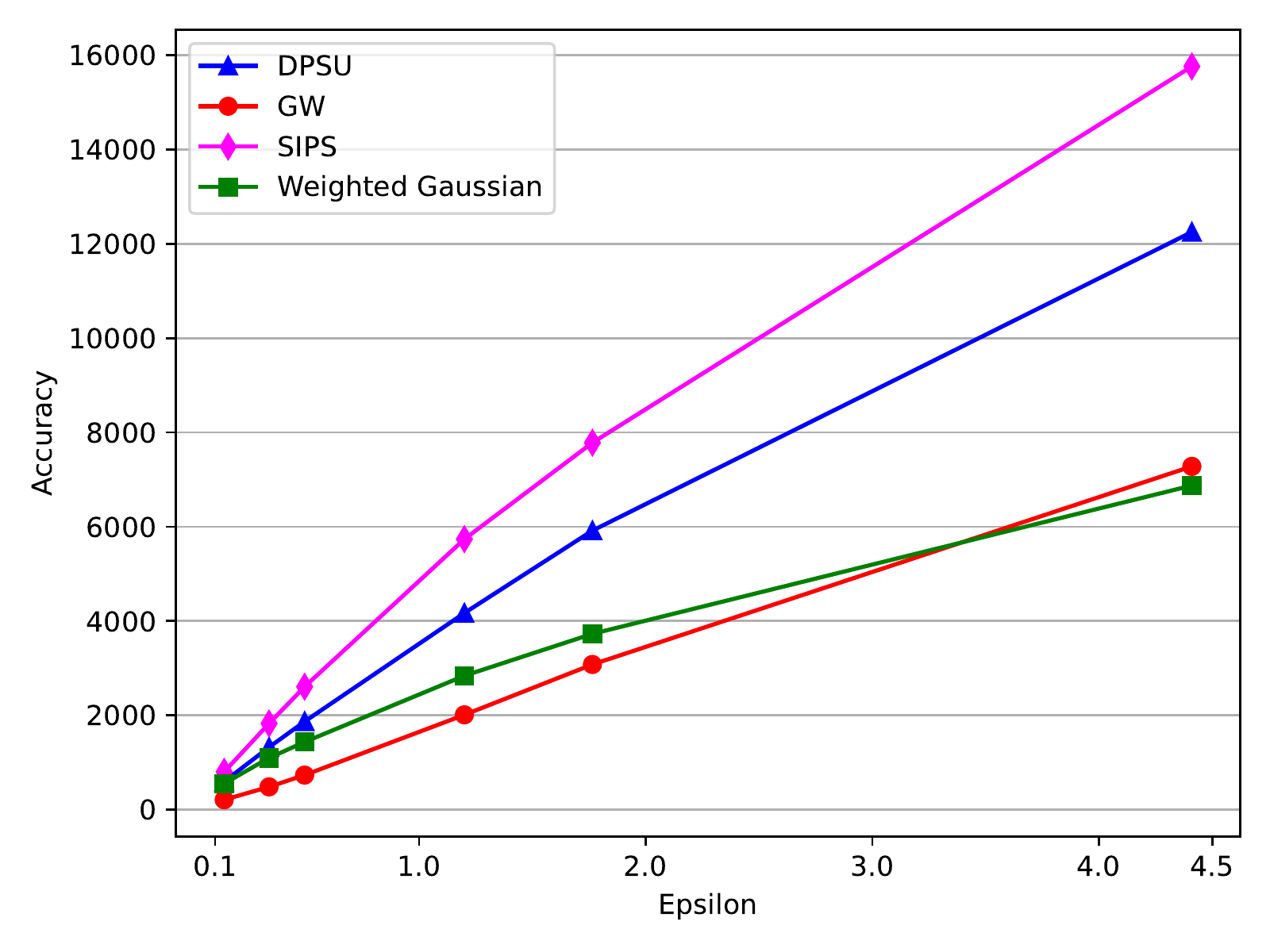}
    \caption{Accuracy on \imdb as a function of varying epsilon for fixed $\delta=10^{-5}$.}
    \label{fig:vary_eps_imdb}
\end{figure}


\subsubsection{Comparing zCDP to DP}
We implement our algorithm to satisfy approximate zCDP to take advantage of its simpler composition properties over standard DP. In~\cite{gopi2020differentially}, the Policy Gaussian algorithm satisfies $(\eps, \delta)$-DP, but the threshold and Gaussian noise can easily be recalibrated to satisfy approximate zCDP instead. We do so in this work to facilitate the comparison to our algorithm. Unfortunately, the GW algorithm from \cite{carvalho2022incorporating} has a bounded $\ell_1$-sensitivity and uses Laplace noise, so it is not easily converted to satisfy approximate zCDP without a significant loss in accuracy. Instead, we use Corollary~\ref{cor:approx_zCDP_to_DP} to choose appropriate $(\eps, \delta)$ parameters for the given $\rho$ and $\delta_{CDP}$. The tables in Appendix~\ref{app:priv_conversions} give the conversions derived from Corollary~\ref{cor:approx_zCDP_to_DP} used for our experiments.

The conversion from approximate zCDP to approximate DP is not exactly tight (meaning the given $(\eps, \delta)$ may be higher than the true privacy guarantee given by the approximate zCDP parameters). Because of the looseness of the conversion, the GW algorithm's accuracy may be slightly inflated when compared to the other algorithms.

Doing the privacy budget accounting using approximate zCDP is not exactly tight either; instead, one could use numerical methods to compute the total privacy budget spent~\cite{meiser2018tight}.
At the time of writing, we could not find a working implementation of such methods that could support the privacy property of Weighted Gaussian.
We do not believe that the results of the experiments would be meaningfully different by doing the privacy analysis this way.
Furthermore, using approximate zCDP has the advantage of being easier to integrate with existing differential privacy software~\cite{tumultanalyticssoftware,tumultanalyticswhitepaper}.

Doing the privacy budget accounting with zCDP is not exactly tight either: numerical methods using

\subsubsection{Selecting hyperparameters}\label{sec:hyperparameters}

Our algorithm has several hyperparameters (aside from the privacy parameters $\rho$ and $\delta$) that need to be set by the data analyst: $\maxcontrib$ the maximum number of items per user, $\privratio$ the ratio between the privacy budget for iteration $i+1$ and $i$, and $\numiters$ the number of iterations in the algorithm. One option is to divide the privacy budget and try several hyperparameter settings and select the setting with the highest accuracy; however, this wastes a lot of privacy budget. We find that the accuracy of DP-SIPS is largely invariant to reasonable settings of the hyperparameters, and we provide general rules of thumb for selecting them.

Figures~\ref{fig:hyperparameters1} and \ref{fig:hyperparameters2} display the accuracy of our algorithm on five data sets with varying \privratio and \maxcontrib, respectively. For the given setting of $\delta$ and $\rho$, the accuracy of DP-SIPS is largely unaffected by different choices of \privratio and \maxcontrib across 5 data sets, and Table~\ref{tab:hyperparameters} shows that the accuracy of DP-SIPS only slightly increases with the number of iterations \numiters on the \reddit data set (for the given settings of the other parameters). \revision{Furthermore, Figure~\ref{fig:hyperparameters1} (privacy ratio $\privratio = 1$) suggests that evenly splitting the budget among the rounds yields lower accuracy than geometrically increasing the budget at each round.}

The results suggest that DP-SIPS has good accuracy for a large range of hyperparameters (aside from the privacy parameters). We recommend using the following settings: $\privratio \in [0.2, 0.4]$, $\numiters \geq 3$, and $\maxcontrib$ set to an overestimate of the true maximum number of per-user contributions. If the true maximum number of contributions per user is public, $\maxcontrib$ should be set to that value.

\begin{table}
\centering
\begin{tabular}{|c|c|}
    \hline
    Iterations $\numiters$ & Partitions Returned \\
    \hline
    \hline
    1 & 5,464\\
    \hline
    2 & 10,041\\
    \hline
    3 & 11,126\\
    \hline
    4 & 11,182\\
    \hline
    5 & 11,541\\
    \hline
    6 & 11,061\\
    \hline
    8 & 11,585\\
    \hline
    10 & 11,637\\
    \hline
\end{tabular}
\caption{SIPS's accuracy as a function of the number of iterations on \reddit with parameters $\rho = 0.1, \delta = 10^{-5}$,  $\privratio = 1/3$, and $\maxcontrib=50$.} \label{tab:hyperparameters}
\end{table}

\begin{figure}
    \centering
    \includegraphics[width=241pt]{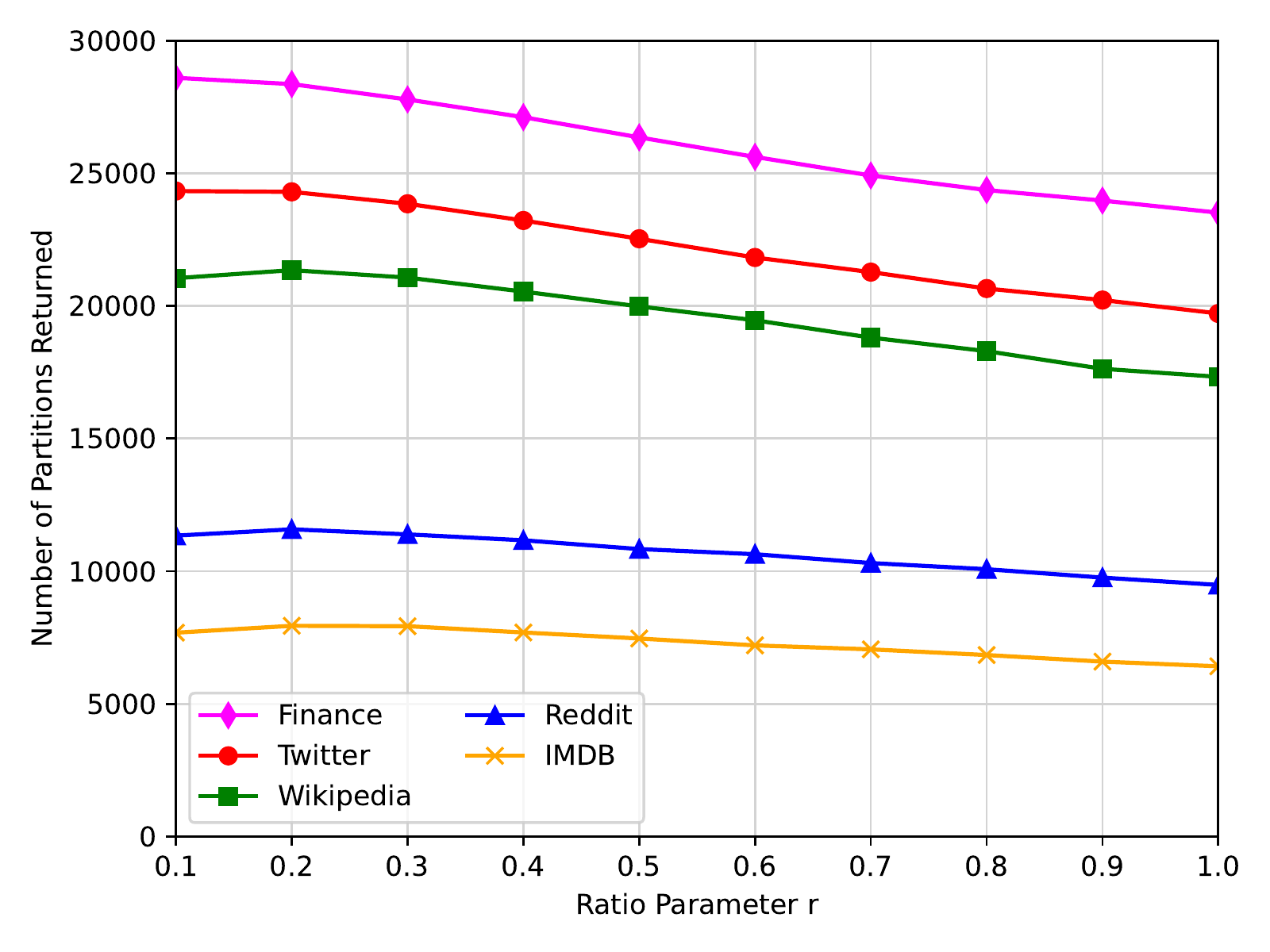}
    \caption{Number of partitions returned by SIPS versus the choice of the privacy ratio parameter, $\privratio$. All other hyperparameters are identical to those listed in Table~\ref{tab:accuracy_benchmarks}. For all 5 data sets, the accuracy is maximized for $r \in [0.1, 0.4]$.}
    \label{fig:hyperparameters1}
\end{figure}

\begin{figure}
    \centering
    \includegraphics[width=241pt]{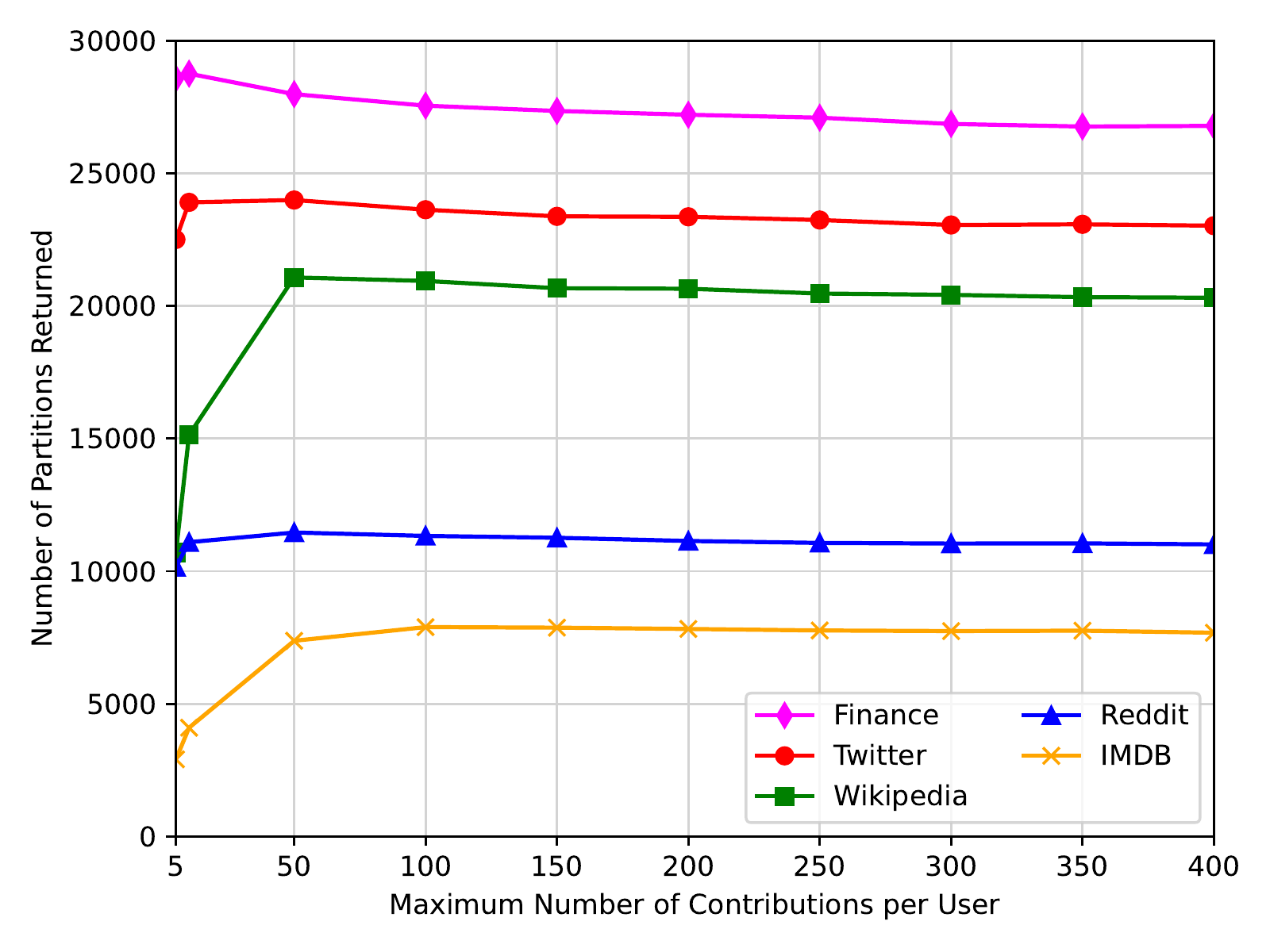}
    \caption{Number of partitions returned by SIPS for varying per-user maximum contribution bounds, $\maxcontrib$. All other hyperparameters are identical to those listed in Table~\ref{tab:accuracy_benchmarks}. On all 5 data sets, increasing $\maxcontrib$ past 100 has almost no affect on accuracy, with maximum accuracy for $\maxcontrib \in [5,150]$. We note that \finance is unusual as each financial headline is quite short.}
    \label{fig:hyperparameters2}
\end{figure}

\subsection{Scalability results}\label{sec:scalability}

\begin{figure*}[h]
    \centering
    \includegraphics[page=1, trim={0.9cm 5cm 2.5cm 1cm},clip, width=469pt]{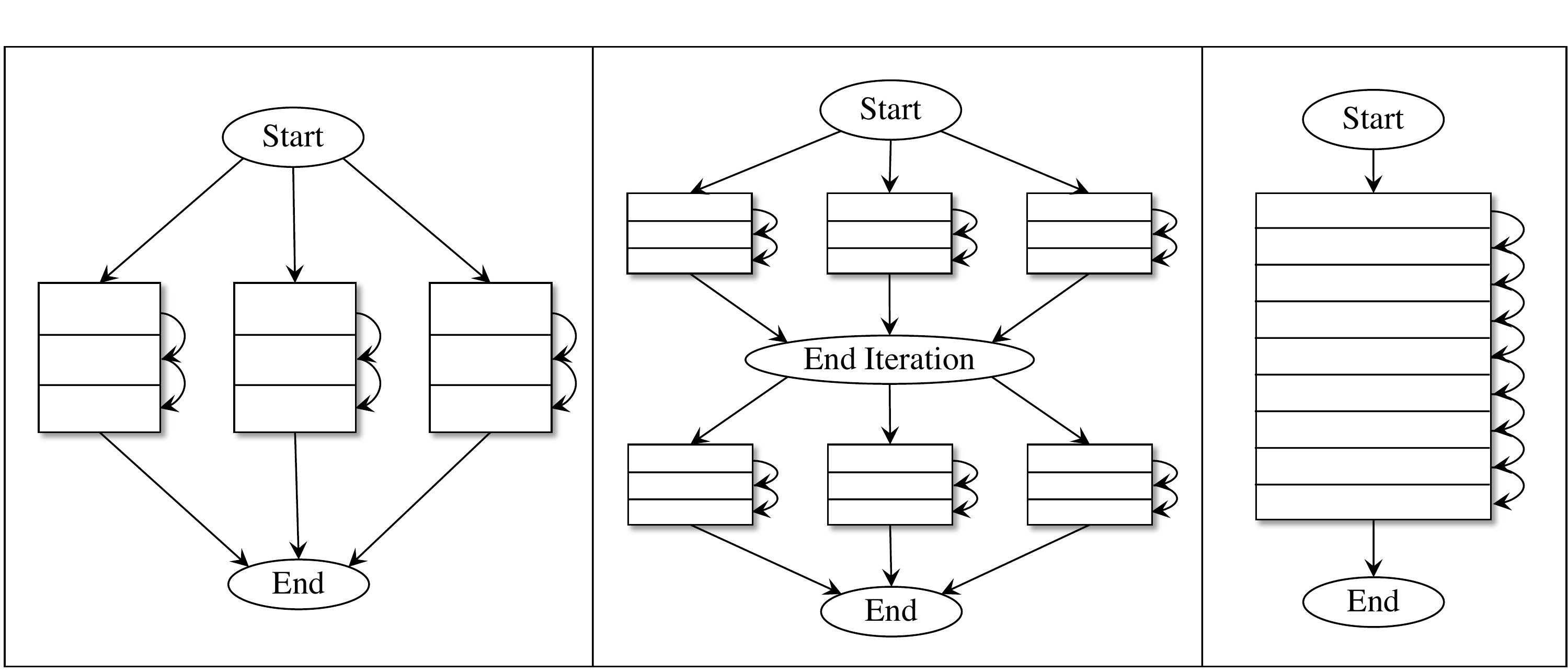}
    \caption{Parallelism diagram for Weighted Gaussian (left), SIPS (middle), and greedy algorithms (right). Arrows represent computational steps while boxes represent the data set. The steps of Weighted Gaussian can run in parallel on parts of the data set while greedy algorithms must run sequentially over the data set. Each iteration of SIPS can run in parallel, but the iterations must be done sequentially.}
    \label{fig:parallelism}
\end{figure*}

We benchmark the algorithms in several ways, and we find that DP-SIPS and Weighted Gaussian scale well with large data sets, while DPSU and GW do not.

For the scalability experiments, we implement all of the algorithms in PySpark to take advantage of parallelism within the algorithms. We then run the algorithms on Amazon Elastic Map Reduce (EMR) clusters, and we tailor the PySpark session settings to the number of cores and memory allocation for both driver and executor nodes to reflect the available resources of the machines.

To benchmark the algorithms, we measure the amount of time required to run the algorithm (after the dataset has already been read in to PySpark) and return the number of partitions that were discovered, in order to ensure PySpark's lazy evaluation executed the algorithm during the timing phase. See Appendix~\ref{app:specs} for information about the machine specifications.

For our first scalability experiment, we run the algorithms on subsets of \synth, a synthetic data set with 80 million users and 4.6 billion items in total, and benchmark the algorithms as the number of users increases. Figure~\ref{fig:synth_runtime} shows the results of this experiment. DP-SIPS and Weighted Gaussian scale well to large data sets while DPSU and GW timeout or run out of memory on subsets with just 2 million and 10 million users, respectively. We ran this experiment on an EMR cluster with 8 nodes, each of size \texttt{m5a.8xlarge} (each node has 32 processor cores and 128 GB of memory).

\begin{figure}
    \centering
    \includegraphics[width=241pt]{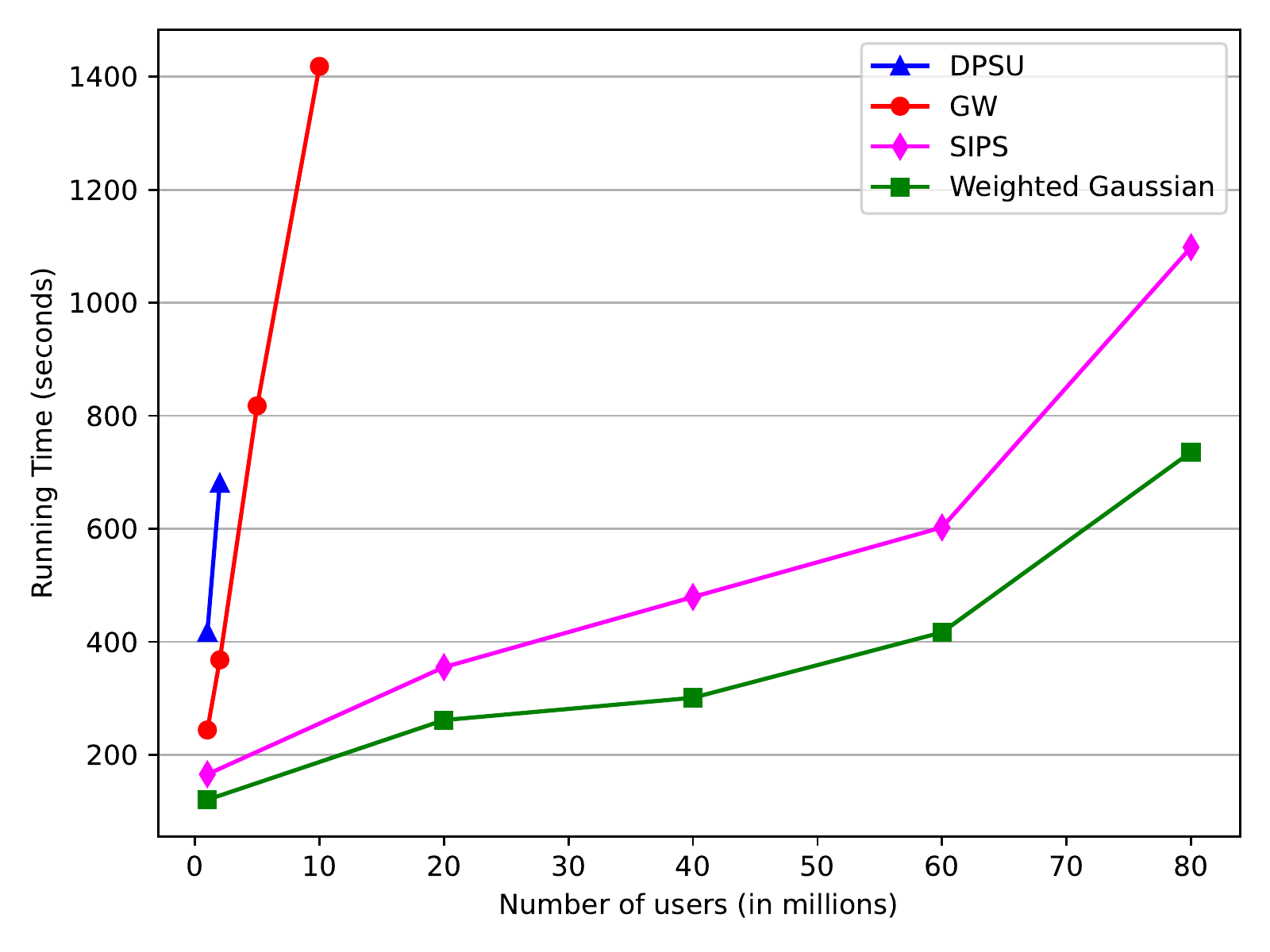}
    \caption{Algorithm runtimes on subsets of \synth of given sizes on a cluster of 8 \texttt{m5a.8xlarge} nodes. The algorithms were run with the same hyperparameters as those listed in Table~\ref{tab:accuracy_benchmarks}. DPSU and GW timed out after subsets with 2 million and 10 million users, respectively.}
    \label{fig:synth_runtime}
\end{figure}

To further investigate the scaling behaviors of the algorithms, we ran more scalability experiments on a smaller data set, \amazon, which consists of product reviews from 4 million users. Figure~\ref{tab:scalability} shows the results from the first experiment, in which we measure the runtimes of the algorithms on clusters with a single master node and varying numbers of core nodes. All nodes are of type m5a.2xlarge, which has 8 processor cores and 32 GB of memory. Figure~\ref{tab:scalability} illustrates the following runtime trends:
\begin{itemize}
    \item The runtimes of Weighted Gaussian and SIPS both decrease as the number of core nodes increases.
    \item The runtimes of DPSU and GW remain approximately the same even as the number of core nodes increases.
\end{itemize}

These results confirm the intuition that Weighted Gaussian and SIPS are both parallelizable and thus scale well with increased cluster sizes, even on large data sets. Additionally, it confirms our observation that since DPSU and GW need to iterate sequentially over each user, increasing the cluster size does little to improve their running times.

\begin{center}
\begin{figure}
    \centering
    \includegraphics[width=241pt]{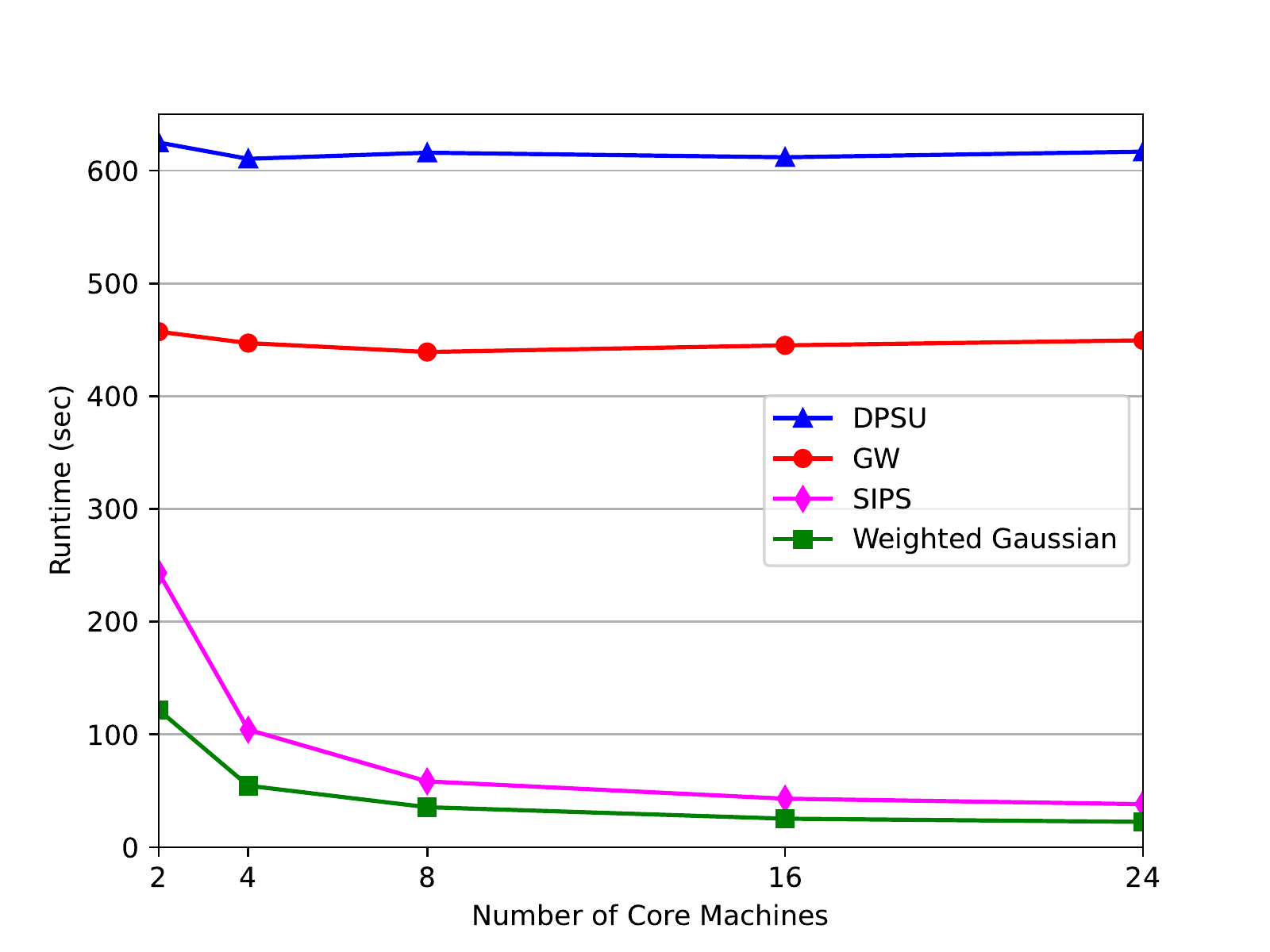}
    \caption{Algorithm runtimes on \amazon dataset with increasing number of m5a.2xlarge cores in the cluster. The algorithms were run with the same hyperparameters as those listed in Table~\ref{tab:accuracy_benchmarks}.}
    \label{tab:scalability}
\end{figure}
\end{center}


We run an additional scalability experiment on \amazon: instead of increasing the number of core nodes, we increase the sizes of the core nodes. We use a single master node and two core nodes, and increase the sizes of all three nodes. Table~\ref{tab:scalability2} shows similar trends as the previous experiment: Weighted Gauss and SIPS scale with increased node sizes while DPSU and GW do not. So, even when the sizes of the machines increase, this makes little difference in the runtimes of DPSU and SIPS.

\begin{center}
\begin{figure}
    \centering
    \includegraphics[width=241pt]{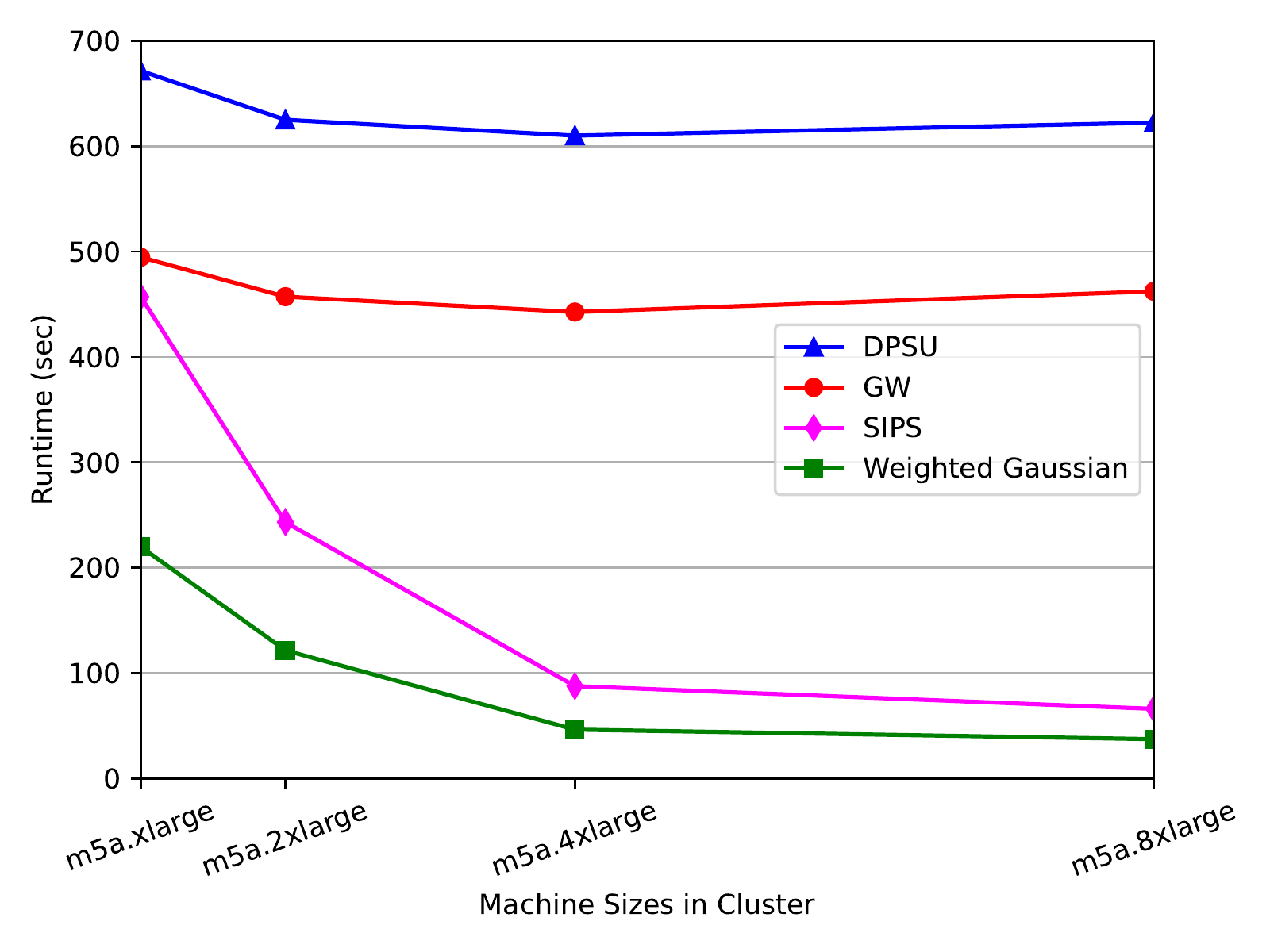}
    \caption{Algorithm runtimes on \amazon dataset with increasing sizes of nodes in a cluster with 1 master node and 2 core nodes. The algorithms were run with the same hyperparameters as those listed in Table~\ref{tab:accuracy_benchmarks}.}
    \label{tab:scalability2}
\end{figure}
\end{center}


Comparing Figures~\ref{tab:scalability} and \ref{tab:scalability2}, we see that for DP-SIPS, the communication overhead between machines is small. Specifically we can compare SIPS's performance on 8 cores in Figure~\ref{tab:scalability}, a cluster with 1 master node and 8 core nodes each of type m5a.2xlarge (8 CPU cores, 32 GB memory) to m5a.8xlarge (32 CPU cores, 128 GB memory) in Figure~\ref{tab:scalability2}, a cluster with 1 master node and 2 core nodes. In both, the core nodes have in total $8 \cdot 8 = 2 \cdot 32$ CPU cores and $8 \cdot 32 = 2 \cdot 128$ GB of memory among them. The runtimes of SIPS in the first case is $58.41$ seconds while in the second case, it is $66.28$ seconds, which suggests that the communication overhead is small. In fact, SIPS runs faster on a cluster of 9 smaller machines than a cluster of 3 larger machines. This discrepancy could be the result of any number of factors in the machine specifications.

Our takeaway is clear: the scalability experiments confirm that DPSU and GW are not suitable for massive data sets and the prohibitively poor run times of DPSU and GW on industrial-scale data sets cannot be overcome by increasing the computational power of the cluster (either in number of machines or even machine size) because the algorithms are not parallelizable on a fundamental level. For larger data sets, SIPS and Weighted Gaussian are the only feasible algorithms and SIPS consistently has improved accuracy over Weighted Gaussian.




\section{Discussion}\label{sec:discussion}
Differentially private partition selection (or set union) is a fundamental problem for many private data analysis tasks. Prior approaches to this problem either suffer from poor accuracy or are prohibitively slow on large data sets because they are designed to sequentially iterate over the users. We present a simple algorithm for differentially private partition selection that achieves accuracy that is comparable to DPSU and runtimes that scale well with increased computational resources. 

\subsection{Unsuccessful Attempts}
One natural question we explored is: can we boost the accuracy of DPSU using our method of iterating several times and removing partitions from the data set that were previously returned? We found that empirically this did not boost the accuracy on the data sets that we tested. Intuitively, this makes sense because the per-user Gaussian update policy in DPSU prevents users from adding weight to items that have already reached the buffered threshold. Such items will very likely be released after the noise addition and thresholding steps. The DPSU approach achieves the same goal as iterating, albeit with more precision due to the greedy nature of the Policy Gaussian per-user updates. 

\subsection{Future work}

This work raises several natural questions. From a practical perspective, one may wonder whether there is a scalable algorithm with higher accuracy than DP-SIPS. Another interesting line of inquiry could consider theoretical guarantees on the number of partitions released, possibly under some distributional assumptions on the data. We believe these lines of inquiry would give important insights into a fundamental problem in private data analysis.
\section{Acknowledgments}

We are thankful to the anonymous reviewers for the helpful feedback. This work was supported by Tumult Labs.

\bibliographystyle{plain}
\bibliography{bib.bib}

\begin{thebibliography}{10}

\bibitem{amin2022plume}
Kareem Amin, Jennifer Gillenwater, Matthew Joseph, Alex Kulesza, and Sergei
  Vassilvitskii.
\newblock Plume: Differential privacy at scale.
\newblock {\em arXiv preprint arXiv:2201.11603}, 2022.

\bibitem{tumultanalyticswhitepaper}
Skye Berghel, Philip Bohannon, Damien Desfontaines, Charles Estes, Sam Haney,
  Luke Hartman, Michael Hay, Ashwin Machanavajjhala, Tom Magerlein, Gerome
  Miklau, Amritha Pai, William Sexton, and Ruchit Shrestha.
\newblock Tumult {{Analytics}}: a robust, easy-to-use, scalable, and expressive
  framework for differential privacy.
\newblock {\em arXiv preprint arXiv:2212.04133}, December 2022.

\bibitem{finance_dataset}
Bot\_Developer.
\newblock Daily financial news for 6000+ stocks.
\newblock
  https://www.kaggle.com/datasets/miguelaenlle/massive-stock-news-analysis-db-for-nlpbacktests.

\bibitem{bun2016concentrated}
Mark Bun and Thomas Steinke.
\newblock Concentrated differential privacy: Simplifications, extensions, and
  lower bounds.
\newblock In {\em Theory of Cryptography Conference}, pages 635--658. Springer,
  2016.

\bibitem{canonne2020discrete}
Cl{\'e}ment~L Canonne, Gautam Kamath, and Thomas Steinke.
\newblock The {{Discrete Gaussian}} for {{Differential Privacy}}.
\newblock In {\em Advances in {{Neural Information Processing Systems}}},
  volume~33, pages 15676--15688. {Curran Associates, Inc.}, 2020.

\bibitem{carvalho2022incorporating}
Ricardo~Silva Carvalho, Ke~Wang, and Lovedeep~Singh Gondara.
\newblock Incorporating item frequency for differentially private set union.
\newblock In {\em Proceedings of the AAAI Conference on Artificial
  Intelligence}, volume~36, pages 9504--9511, 2022.

\bibitem{dwork2006our}
Cynthia Dwork, Krishnaram Kenthapadi, Frank McSherry, Ilya Mironov, and Moni
  Naor.
\newblock Our data, ourselves: Privacy via distributed noise generation.
\newblock In {\em Annual international conference on the theory and
  applications of cryptographic techniques}, pages 486--503. Springer, 2006.

\bibitem{dwork2006calibrating}
Cynthia Dwork, Frank McSherry, Kobbi Nissim, and Adam Smith.
\newblock Calibrating noise to sensitivity in private data analysis.
\newblock In {\em Theory of cryptography conference}, pages 265--284. Springer,
  2006.

\bibitem{privacyonbeam}
Google.
\newblock {Privacy on Beam}.
\newblock
  \url{https://github.com/google/differential-privacy/tree/main/privacy-on-beam},
  November 2022.

\bibitem{pipelinedp}
Google and {OpenMined}.
\newblock {PipelineDP}.
\newblock \url{https://pipelinedp.io}, November 2022.

\bibitem{gopi2020differentially}
Sivakanth Gopi, Pankaj Gulhane, Janardhan Kulkarni, Judy~Hanwen Shen, Milad
  Shokouhi, and Sergey Yekhanin.
\newblock Differentially private set union.
\newblock In {\em International Conference on Machine Learning}, pages
  3627--3636. PMLR, 2020.

\bibitem{amazon_dataset}
Kritanjali Jain.
\newblock Amazon reviews.
\newblock https://www.kaggle.com/datasets/ kritanjalijain/amazon-reviews.

\bibitem{kamath2019privately}
Gautam Kamath, Jerry Li, Vikrant Singhal, and Jonathan Ullman.
\newblock Privately learning high-dimensional distributions.
\newblock In {\em Conference on Learning Theory}, pages 1853--1902. PMLR, 2019.

\bibitem{korolova2009releasing}
Aleksandra Korolova, Krishnaram Kenthapadi, Nina Mishra, and Alexandros
  Ntoulas.
\newblock Releasing search queries and clicks privately.
\newblock In {\em Proceedings of the 18th international conference on World
  wide web}, pages 171--180, 2009.

\bibitem{tumultanalyticssoftware}
Tumult Labs.
\newblock Tumult {{Analytics}}.
\newblock \url{https://tmlt.dev}, December 2022.

\bibitem{meiser2018tight}
Sebastian Meiser and Esfandiar Mohammadi.
\newblock Tight on budget? tight bounds for r-fold approximate differential
  privacy.
\newblock In {\em Proceedings of the 2018 ACM SIGSAC Conference on Computer and
  Communications Security}, pages 247--264, 2018.

\bibitem{imdb_dataset}
Lakshmipathi N.
\newblock Imdb dataset of 50k movie reviews.
\newblock
  https://www.kaggle.com/datasets/lakshmi25npathi/imdb-dataset-of-50k-movie-reviews.

\bibitem{reddit_dataset}
Judy~Hanwen Shen.
\newblock Ask reddit.
\newblock
  https://github.com/heyyjudes/differentially-private-set-union/tree/ea7b39285dace35cc9e9029692802759f3e1c8e8/data.

\bibitem{twitter_dataset}
Thought Vector and Stuart Axelbrooke.
\newblock Customer support on twitter.
\newblock
  https://www.kaggle.com/datasets/thoughtvector/customer-support-on-twitter.

\bibitem{wiki_dataset}
Mark Wijkhuizen.
\newblock Simple/normal wikipedia abstracts v1.
\newblock
  https://www.kaggle.com/datasets/markwijkhuizen/simplenormal-wikipedia-abstracts-v1.

\bibitem{wilson2020differentially}
Royce~J Wilson, Celia~Yuxin Zhang, William Lam, Damien Desfontaines, Daniel
  Simmons-Marengo, and Bryant Gipson.
\newblock Differentially private {SQL} with bounded user contribution.
\newblock {\em Proceedings on Privacy Enhancing Technologies}, 2:230--250,
  2020.

\end{thebibliography}

\appendix
\section{Proof of Theorem~\ref{thm:naive_priv}}\label{app:naive_priv}

\begin{proof}[Proof of Theorem~\ref{thm:naive_priv}]
Since the algorithm is simply computing a stable histogram, the proof follows from standard arguments.

Let us denote Algorithm~\ref{alg:naive} by $\mech$.
Fix any $\rho>0$, any $\delta \in (0,1)$, and any $\maxcontrib\in \mathbb{N}$.
Fix two neighboring data sets $x = (W_1, \ldots, W_N)$ and $x'$, where $x'$ contains one extra user set $W^*$.
Let $E$ be the event that $\mech(x') \subseteq \cup_{i\in [N]} \overline{W_i}$; that is, the partitions returned for data set $x'$ are in the support of $\mech(x)$.
We will first argue that, conditioned on event $E$, the output distributions of $\mech(x)$ and $\mech(x')$ satisfy pure zCDP (Definition~\ref{def:pure_zCDP}).
Then we will argue that, because of the thresholding step, event $E$ occurs with probability at least $1-\delta$. 

If we condition both $\mech(x)$ and $\mech(x')$ on event $E$, then by the Gaussian mechanism (Definition~\ref{def:gauss_mech}) and post-processing (Lemma~\ref{lem:comp}), Algorithm~\ref{alg:naive} satisfies pure zCDP (Definition~\ref{def:pure_zCDP}).

Now, it remains to show that $\Pr[E]\geq 1-\delta$. First, note that each user contributes to most \maxcontrib items in the histogram, and further note that event $E^c$ occurs when at least one item from $W^*$ that is not in $x$ is released. Let $W'$ denote the set of items in $W^*$ that do not appear in $x$. Then,

\begin{align*}
    \Pr[E]
    & = \Pr[\forall u \in W', u \notin \mech(x)]\\
    & = \Pr[\cap_{u\in W'} \hat{H}[u] \leq \thresh]\\
    & = \Pr[\cap_{u\in W'} H[u] + Z_u \leq \thresh] \quad \text{For $Z_u \sim \N(0, 1/2\rho)$}\\
    & = \prod_{u\in W'} \Pr[H[u] + Z_u \leq \thresh] \quad \text{By independence of $Z_u$'s}\\
    & = \prod_{u\in W'} \Pr\left[\frac{1}{\sqrt{|W'|}} + Z_u \leq \thresh\right]\\
    & = \left(\Pr\left[\frac{1}{\sqrt{|W'|}} + Z_u \leq \thresh\right]\right)^{|W'|} \quad \text{Since $Z_u$'s are i.i.d.}\\
    & \geq \min_{k \in [\maxcontrib]} \left\{ \left(\Pr\left[\frac{1}{\sqrt{k}} + Z_u \leq \thresh\right]\right)^{k}\right \} \quad \text{Since $|W'| \leq \maxcontrib$}\\
    & \geq 1-\delta \quad \text{By definition of $\thresh$}.
\end{align*}

\revision{
Therefore, $\mech$ is $\delta$-approximate $\rho$-zCDP.
}
\end{proof}


\section{Privacy Parameter Conversions}\label{app:priv_conversions}
Tables~\ref{tab:vary_eps} and \ref{tab:vary_delta} give the conversions between zCDP and DP that we use for the experiments, as well as the values of $\alpha$ that are used in Corollary~\ref{cor:approx_zCDP_to_DP} to do the conversions.

\begin{table}[H]
    \centering
    \begin{tabular}{|c|c|c|c|c|}
        \hline
         $\rho$ & $\delta_{CDP}$ & $\eps$ & $\delta_{DP}$ & $\alpha$ \\
         \hline
         0.001 & $1\times 10^{-5}$ & 0.14 & $5.00 \times 10^{-5}$ & 77.033 \\
         0.005 & $1\times 10^{-5}$ & 0.338 & $5.08 \times 10^{-5}$ & 37.037 \\
         0.01 & $1\times 10^{-5}$ & 0.495 & $4.99 \times 10^{-5}$ & 27.128 \\
         0.05 & $1\times 10^{-5}$ & 1.2 & $4.99 \times 10^{-5}$ & 13.283 \\
         0.1 & $1\times 10^{-5}$ & 1.765 & $4.96 \times 10^{-5}$ & 9.86 \\
         0.5 & $1\times 10^{-5}$ & 4.41 & $4.90 \times 10^{-5}$ & 5.127 \\
         \hline
    \end{tabular}
    \caption{Values of $\rho, \delta_{CDP}, \eps,$ and $\delta_{DP}$ used for experiments with fixed $\delta_{DP}$ and varying $\eps$.}
    \label{tab:vary_eps}
\end{table}

\begin{table}[H]
    \centering
    \begin{tabular}{|c|c|c|c|c|}
        \hline
         $\rho$ & $\delta_{CDP}$ & $\eps$ & $\delta_{DP}$ & $\alpha$ \\
         \hline
         0.005 & $1\times 10^{-9}$ & 0.62 & $1.04 \times 10^{-9}$ & 64.073 \\
         0.0055 & $1\times 10^{-8}$ & 0.62 & $1.02 \times 10^{-8}$ & 58.443 \\
         0.006 & $1\times 10^{-7}$ & 0.62 & $1.01 \times 10^{-7}$ & 53.732 \\
         0.007 & $1\times 10^{-6}$ & 0.62 & $1.01 \times 10^{-6}$ & 46.334 \\
         0.0083 & $1\times 10^{-5}$ & 0.62 & $1.01 \times 10^{-5}$ & 39.398 \\
         0.01 & $1\times 10^{-4}$ & 0.62 & $1.01 \times 10^{-4}$ & 33.037 \\
         0.013 & $1\times 10^{-3}$ & 0.62 & $1.01 \times 10^{-3}$ & 25.863 \\
         \hline
    \end{tabular}
    \caption{Values of $\rho, \delta_{CDP}, \eps, $ and $\delta_{DP}$ used for experiments with fixed $\eps$ and varying $\delta_{DP}$.}
    \label{tab:vary_delta}
\end{table}
\section{Amazon EMR Cluster Specifications}\label{app:specs}

Table~\ref{tab:EMR_specs} shows the number of CPU cores and amount of memory available to each type of machine on Amazon EMR that we use in our experiments. 

\begin{table}[H] 
    \centering
    \begin{tabular}{|c|c|c|}
    \hline
        Machine Size & Cores & Memory (GB)\\
        \hline
        m5a.xlarge & 4 & 16 \\
        m5a.2xlarge & 8 & 32 \\
        m5a.4xlarge & 16 & 64\\
        m5a.8xlarge & 32 & 128\\
        \hline
    \end{tabular}
    \caption{Number of cores and memory per EMR machine of each size.}
    \label{tab:EMR_specs}
\end{table}

\end{document}